\documentclass[11pt]{llncs}
\usepackage{microtype}
\usepackage{makeidx}
\usepackage{graphicx}
\usepackage{xcolor}
\usepackage{ifthen}
\usepackage{microtype}
\usepackage{boxedminipage}
\usepackage{amsmath}
\usepackage{amsfonts}
\usepackage{amssymb}
\usepackage{authblk}
\usepackage{breakcites}
\usepackage{chemscheme}
\usepackage{hyperref}
\usepackage{framed}
\usepackage{float}
\usepackage{breqn}
\usepackage{array, makecell}

\newcommand{\lv}[1]{}

\newcommand{\pr}{\mathbf{Pr}}
\newcommand{\eps}{{\varepsilon}}
\newcommand{\etal}{{\it et al. }}

\newcommand{\nl}{\bot}
\newcommand{\veps}{\varepsilon}
\newcommand{\poly}{{\tt poly}}

\renewcommand{\L}{{\cal L}}
\newcommand{\ti}{{\bar i}}
\newcommand{\ostp}[1]{{O_{#1}^{'}}}

\newcommand{\X}{\mathcal{X}}

\newcommand{\R}{\mathbb{R}}

\newcommand{\Xst}{X^{\star}}

\newcommand{\cst}{c^{\star}}
\newcommand{\GC}{{\tt $t$-GoodCenters }}
\newcommand{\C}{\mathbb{C}}
\newcommand{\OPT}{OPT^{\star}}

\newcommand{\norm}[1]{\lVert #1 \rVert}

\newtheorem{fact}{Fact}

\allowdisplaybreaks

\graphicspath{{./Figures/}}

\begin{document}

\title{Streaming {PTAS} for Constrained $k$-Means}
%
%
\author{Dishant Goyal \and Ragesh Jaiswal \and Amit Kumar}
%
%
%
\institute{
Department of Computer Science and Engineering, \\
Indian Institute of Technology Delhi.\thanks{Email addresses: \email{\{Dishant.Goyal, rjaiswal, amitk\}@cse.iitd.ac.in}}
}
{\def\addcontentsline#1#2#3{}\maketitle}

\begin{abstract}
We generalise the results of Bhattacharya \etal~\cite{bjk} for the {\em list-$k$-means} problem defined as -- for a (unknown) partition $X_1, ..., X_k$ of the dataset $X \subseteq \R^d$, find a {\em list} of $k$-center sets (each element in the list is a set of $k$ centers) such that at least one of $k$-center sets $\{c_1, ..., c_k\}$ in the list gives an $(1+\veps)$-approximation with respect to the cost function $\min_{\textrm{permutation } \pi} \left[ \sum_{i=1}^{k} \sum_{x \in X_i} ||x - c_{\pi(i)}||^2 \right]$. 
The list-$k$-means problem is important for the constrained $k$-means problem since algorithms for the former can be converted to {PTAS} for various versions of the latter. 
The algorithm for the list-$k$-means problem by Bhattacharya \etal is a $D^2$-sampling based algorithm that runs in $k$ iterations exploring a tree of size $\left(\frac{k}{\veps}\right)^{O(\frac{k}{\veps})}$.
Under the assumption that a constant factor solution  is available for the (classical or unconstrained) $k$-means problem, we generalise the algorithm of Bhattacharya \etal in two ways -- (i) the algorithm runs in a single iteration, and (ii) for any fixed set $X_{j_1}, ..., X_{j_t}$ of $t \leq k$ clusters, the algorithm produces a list of $(\frac{k}{\veps})^{O(\frac{t}{\veps})}$ $t$-center sets such that (w.h.p.) at least one of them is good for $X_{j_1}, ..., X_{j_t}$.
Following are the consequences of our generalisations:
\begin{enumerate}
\item {\em Streaming algorithm}: The $D^2$-sampling algorithm running in a single iteration allows us to design a 2-pass, logspace streaming algorithm for the list-$k$-means problem.This can be converted to a 4-pass, logspace streaming {PTAS} for various constrained versions of the $k$-means problem.

\item {\em Faster PTAS under stability}: The second generalisation is useful in $k$-means clustering scenarios where finding good centers becomes easy once good centers for a few ``bad" clusters have been chosen. One such scenario is clustering under stability of Awasthi \etal~\cite{abs10} where the number of such bad clusters is a constant. 
Using the above idea, we significantly improve the running time of their algorithm from $O(dn^3) (k \log{n})^{\poly(\frac{1}{\beta}, \frac{1}{\veps})}$ to $O \left(dn^3  \left(\frac{k}{\veps} \right)^{O(\frac{1}{\beta \veps^2})} \right)$.

\item {\em Parallel Algorithm}: The algorithm of Bhattacharya \etal~\cite{bjk} is highly parallelizable except for an iteration of size $k$. Our single iteration algorithm allows us to convert a constant factor approximate solution to a $(1+\veps)$-factor approximate solution in fast parallel time.
\end{enumerate}
\end{abstract}

\tableofcontents



\newpage
\pagestyle{plain}
\setcounter{page}{1}

\section{Introduction}
Clustering is one of the most important tools for data analysis and the $k$-means clustering problem is the most prominent mathematical formulations of clustering. 
The goal of clustering is to partition data objects into groups, called {\em clusters}, such that similar objects are in the same cluster and dissimilar ones are in different clusters.
Defining the clustering problem mathematically requires us to quantify the notion of similarity/dissimilarity and there are various ways of doing this.
Given that in most contexts data objects can be represented as vectors in $\R^d$, a natural notion of distance between data points is the squared Euclidean distance and this gives rise to the $k$-means problem.
\begin{quote}
\underline{\bf The $k$-means problem}: Given a dataset $X \subset \R^d$ and a positive integer $k$, find a set $C \subset \R^d$ of $k$ points, called {\em centers}, such that the following cost function gets minimised:
$$\Phi(C, X) \equiv \sum_{x \in X} \min_{c \in C}{||x - c||^2}.$$
\end{quote}
The $k$-means problem has been widely studied by both theoreticians and practitioners and is quite uniquely placed in the computer science research literature. 
The theoretical worst-case analysis properties of the $k$-means problem is fairly well understood.
The problem has been shown to be $\mathsf{NP}$-hard~\cite{das08,mnv12,V09} and $\mathsf{APX}$-hard~\cite{acks15}.
A lot of work has been done on obtaining efficient constant approximation algorithms for this problem (e.g., \cite{KanungoMNPSW02,ahmadian17}). 
However, this is not the main focus of this work. 
In this work, we disucss approximation schemes for the $k$-means problem and its variants. 
Approximation schemes are family of algorithms $\{A\}_{\veps}$ that give $(1+\veps)$-approximation guarantee.

Given the hardness of approximation result~\cite{acks15}, it is known that a {\em Polynomial Time Approximation Scheme (PTAS)} is not possible unless $\mathsf{P} = \mathsf{NP}$.
However, there are efficient approximation schemes when at least one of $k, d$ is not part of the input (and hence assumed to be a fixed constant). 
The work on approximation schemes for the $k$-means problem can be split into two categories where one consists of algorithms under the assumption that $k$ is a constant while the other with $d$ as a constant.
Assuming $k$ is to be a constant, there are various PTAS~\cite{kss,FeldmanMS07,jks,jky15} with running time $O(nd \cdot 2^{\tilde{O}(\frac{k}{\veps})})$.\footnote{The multiplicative factor of $nd$ can be changed to an additive factor using useful data analysis tools and techniques such as {\em coresets}~\cite{FeldmanMS07} and {\em dimensionality reduction}~\cite{Linial1995}.}
Note that the running time has a dependence on $2^k$. This is nicely supported by a conditional lower bound result~\cite{abjk} that says that under the {\em Exponential Time Hypothesis (ETH)} any  approximation algorithm (beyond a fixed approximation factor) that runs in time polynomial in $n$ and $d$ will have a running time dependence of at least $2^k$.
On the other hand, PTAS based on the assumption that $d$ is a constant form another  line of research culminating in the work of Addad \etal~\cite{addad16} and Friggstad \etal~\cite{friggstad16} who gave a local search based PTAS with running time dependence on $d$ of the form $(\frac{k}{\veps})^{\zeta}$ where $\zeta = \frac{d^{O(d)}}{\veps^{O(\frac{d}{\veps})}}$.
The recent work of Makarychev \etal~\cite{mmr18} nicely consolidates the two lines of work by showing that the cost of the optimal $k$-means solution is preserved up to a factor of $(1+\veps)$ under a projection onto a random $O \left( \frac{\log{(k/\veps)}}{\veps^2} \right)$-dimensional subspace.

The $k$-means problem nicely models the {\em locality} requirement of clustering. 
That is, similar (or closely located points) should be in the same cluster and dissimilar (or far-away points) should be in different clusters. 
However, in many different clustering contexts in machine learning and data mining, locality is not the only desired clustering property. 
There are other constraints in addition to the to the locality requirement. 
For example, one requirement is that the clusters should be balanced or in other words contain roughly equal number of points. 
Modelling such requirements within the framework of the $k$-means problem gives rise to something known as a {\em constrained $k$-means problem}.
A constrained $k$-means problem can be modelled as follows: Let $\mathbb{C}$ denote the set of $k$-clusterings that satisfy the relevant constraint. Then the goal is to find a clustering $\X = \{X_1, ..., X_k\}$ of the dataset $X \subset \R^d$ such that the clustering is in $\mathbb{C}$ and the following cost function is minimised:
\[
\Delta(\X) \equiv \sum_{i=1}^{k} \Delta(X_i), \textrm{ where } \Delta(X_i) \equiv \Phi(\mu(X_i), X_i) \textrm{ and } \mu(X_i) \equiv \frac{\sum_{x \in X_i} x}{|X_i|}.
\]
Note that  $\mu(X_i)$ is the {\em centroid} of the data points $X_i$.\footnote{It can be easily shown that the centroid gives the best 1-means cost for any dataset and so $\Delta(X_i)$ denotes the optimal 1-means cost of dataset $X_i$.}
The above formulation in terms of the feasible clusterings $\mathbb{C}$ is an attempt to give a unified framework for considering different variations of the constrained clustering problem. 
The issue with such an attempt is how to concisely represent the set of feasible clusterings $\mathbb{C}$.
This issue was addressed in the nice work of Ding and Xu~\cite{dx15} who gave a unified framework for considering constrained versions of the $k$-means problem.
For every constrained version, instead of defining $\mathbb{C}$ they define a {\em partition algorithm} $\mathcal{P}^{\C}$ which when given a set of $k$ centers $\{c_1, ..., c_k\}$ outputs a feasible clustering $\{X_1, ..., X_k\}$ (i.e., a clustering in $\mathbb{C}$) that minimises the cost $\sum_{i=1}^{k} \Phi(\{c_i\}, X_i)$. 
They give efficient partition algorithms for a variety of constrained $k$-means problems. 
These problems and their description are given in Table~\ref{table:1}.
Note that the partition algorithm for the $k$-means problem (i.e., the classical unconstrained version) is simply the {\em Voronoi partitioning} algorithm.

\begin{table}
\centering
\begin{tabular}{|l|l|l|}
\hline
\# & {\bf Problem} & {\bf Description} \\ \hline
1. & \makecell[l]{$r$-gather $k$-means clustering \\ $(r, k)$-{\tt GMeans}} & \makecell[l]{Find clustering $\X = \{X_1, ..., X_k\}$ with minimum $\Delta(\X)$ \\such that for all $i$, $|X_i| \geq r$} \\ \hline
2. & \makecell[l]{$r$-Capacity $k$-means clustering \\ $(r, k)$-{\tt CaMeans}} & \makecell[l]{Find clustering $\X = \{X_1, ..., X_k\}$ with minimum $\Delta(\X)$ \\such that for all $i$, $|X_i| \leq r$} \\ \hline
3. & \makecell[l]{$l$-Diversity $k$-means clustering \\ $(l, k)$-{\tt DMeans}} & \makecell[l]{Given that every data point has an associated colour, \\find a clustering $\X = \{X_1, ..., X_k\}$  with minimum $\Delta(\X)$ \\such that for all $i$, the fraction of points sharing the \\same colour inside $X_i$ is $\leq \frac{1}{l}$} \\ \hline
4. & \makecell[l]{Chromatic $k$-means clustering \\ $k$-{\tt ChMeans}} & \makecell[l]{Given that every data point has an associated colour,  \\find a clustering $\X = \{X_1, ..., X_k\}$ with minimum $\Delta(\X)$ \\such that for all $i$, $X_i$ should not have more than two \\points with the same colour.} \\ \hline
5. & \makecell[l]{Fault tolerant $k$-means clustering \\ $(l, k)$-{\tt FMeans}} & \makecell[l]{Find clustering $\X = \{X_1, ..., X_k\}$ such that \\the sum of squared distances of the points to the $l$ nearest \\centers out of $\{\mu(X_1), ..., \mu(X_k)\}$, is minimised.} \\ \hline
6. & \makecell[l]{Semi-supervised $k$-means clustering \\ $k$-{\tt SMeans}} & \makecell[l]{Given a target clustering $\X' = \{X_1', ..., X_k'\}$ and constant $\alpha$ \\find a clustering
$\X = \{X_1, ..., X_k\}$ such that the cost \\$\alpha \cdot \Delta(\X) + (1-\alpha) \cdot Dist(\X', \X)$ is minimised. \\$Dist$ denotes the set-difference distance.} \\ \hline
\end{tabular}
\caption{Constrained $k$-means problems with efficient partition algorithm (see Section~4 in \cite{dx15}).}\label{table:1}
\end{table}

Efficient partition algorithms allows us to design PTAS in the following manner: Let $\X = \{X_1, ..., X_k\}$ be an optimal clustering for some constrained $k$-means problem with optimal cost $OPT = \Delta(\X) = \sum_{i=1}^{k} \Delta(X_i)$. 
Suppose in some way, we are able to find a $k$-center set $\{c_1, ..., c_k\}$ such that
\[
\min_{\textrm{permutation } \pi} \left[ \sum_{i=1}^{k} \sum_{x \in X_i} ||x - c_{\pi(i)}||^2 \right] \leq (1 + \veps) \cdot OPT.
\]
Then we can use the partition algorithm to find a clustering $\bar{\X} = \{\bar{X}_1, ..., \bar{X}_k\}$ such that $\Delta(\bar{\X}) \leq (1+\veps) \cdot OPT$.
It turns out that even though producing a single such $k$-center set may not be possible, producing a {\em list} of such $k$-center sets is possible. 
Using the partition algorithm to find the clustering with least cost from the list will give us a $(1+\veps)$-approximate solution.
This is the main idea used for designing PTAS by Ding and Xu~\cite{dx15} and Bhattacharya \etal~\cite{bjk}.
Bhattacharya \etal~\cite{bjk} gave quantitative improvements over the results of Ding and Xu in terms of the list size.
They also formally defined the {\em list-$k$-means} problem that is a natural problem in the context of the above discussion.\footnote{Note that Ding and Xu implicitly gave an algorithm for the list-$k$-means problem in their work~\cite{dx15} without naming it so.}
One of the main focus of discussion of this paper will be the list-$k$-means problem. 
So, let us first define the problem formally.

\begin{quote}
{\bf List-$k$-means}: Let $X \subset \R^d$ be the dataset and let $\X = \{X_1, ..., X_k\}$ be an arbitrary clustering of dataset $X$. 
Given $X$, positive integer $k$, and error parameter $\veps > 0$, find a {\em list} of $k$-center sets such that (whp\footnote{We use whp as an abbreviation for ``with high probability".}) at least one of the sets gives $(1+\veps)$-approximation with respect to the cost function: 
\[
\psi(\{c_1, ..., c_k\}, \X) \equiv \min_{\textrm{permutation } \pi} \left[ \sum_{i=1}^{k} \sum_{x \in X_i} ||x - c_{\pi(i)}||^2 \right].
\]
\end{quote}
Bhattacharya \etal~\cite{bjk} gave a lower bound on the list size using a counting argument and a closely matching upper bound using a {\em $D^2$-sampling} based approach.
$D^2$-sampling is a simple idea that is very useful in the context of the $k$-means/median clustering problems.
Here the centers are sampled from the given dataset in successive iterations where the probability of a point getting sampled as the center in an iteration is proportional to the squared distance of this point to the nearest center out of the centers already chosen in the previous iterations.
Before discussing the algorithm for the list-$k$-means problem, let us first make sure that the relevance of this problem in the context of the constrained $k$-means problems is well understood.
Indeed, given any constrained $k$-means clustering problem with feasible clusterings $\C$ and partition algorithm $\mathcal{P}^{\C}$, one can obtain a $(1+\veps)$-approximate solution by first running an algorithm for the list $k$-means problem (where the unknown clustering is any optimal clustering for the constrained $k$-means problem) to obtain a list $\L$ and then use the partition algorithm $\mathcal{P}^{\C}$ to pick the minimum cost clustering from $\L$.
From the previous discussion, it should be clear that this will give is a $(1+\veps)$-approximate solution (whp). 
Let us now discuss the $D^2$-sampling based algorithm for the list-$k$-means problem.

Bhattacharya \etal~\cite{bjk} gave an algorithm for the list-$k$-means problem with list size $|\L| = (\frac{k}{\veps})^{O(\frac{k}{\veps})}$ and running time $O(nd |\L|)$.
Their algorithm explores a rooted tree of size $(\frac{k}{\veps})^{O(\frac{k}{\veps})}$ and depth $k$ where the degree of every non-leaf vertex is 
$(\frac{k}{\veps})^{O(\frac{1}{\veps})}$. 
Every node in this tree has an associated center and the path from root to a leaf node gives one of the $k$-center sets for the output list.
Let $v$ be an internal node at depth $i$. 
The path from root to $v$ defines $i$ centers $C_v$ and their algorithm extends these $i$ centers to $(i+1)$ centers by $D^2$-sampling $\poly(\frac{k}{\veps})$ points w.r.t. $C_v$ and considering the centroids of all possible subsets of size $O(\frac{1}{\veps})$ of the sampled points plus copies of centers in $C_v$.\footnote{$D^2$-sampling w.r.t. a center set $C$ implies sampling from the dataset $X$ using a distribution where the probability of sampling point $x$ is proportional to $\min_{c \in C}{||x - c||^2}$. In the case $C = \emptyset$, $D^2$-sampling is the same as uniform sampling.}
This defines the $(\frac{k}{\veps})^{O(\frac{1}{\veps})}$ children of $v$ that are further explored subsequently.
In their analysis, they showed that for every node $v$, there is always (whp) a child of $v$ that is a good center for one of the clusters for which none of the centers in $C_v$ is good.

Note that the algorithm of Bhattacharya \etal\cite{bjk} in the previous paragraph has an unavoidable iteration of depth $k$ since their analysis works only when the centers are picked {\em one-by-one} in $k$ iterations.
We circumvent this inherent restriction by using a constant factor approximate solution $C$ to the  $k$-means problem (i.e., the unconstrained $k$-means problem) for the given dataset $X$.
That is, $\Phi(C, X) \leq \alpha \cdot \OPT$, where $\OPT$ denotes the optimal $k$-means cost. 
Note that there are a number of constant factor approximation algorithms available for the $k$-means problem. 
So, this assumption is not restrictive at all.
We can even further relax the assumption by noting that an $(O(1), O(1))$ bi-criteria approximate solution $C$ is sufficient. This means that $|C| = O(k)$ and $\Phi(C, X) \leq \alpha \cdot \OPT$.
There are bi-criteria approximation algorithms available for the $k$-means problem.
For example, there is a simple $O(nkd)$ bi-criteria approximation algorithm based on $D^2$-sampling that just samples $O(k)$ points (using $D^2$-sampling) and it has been shown~\cite{adk09} that the set of centers obtained gives a constant approximation with high probability.
Under the assumption that such a constant factor solution $C$ is available, we generalise the $D^2$-sampling based algorithm of Bhattacharya \etal~\cite{bjk} in the following two ways:
\begin{enumerate}
\item We consider the case where we may not need to find good centers for {\em all} clusters but for  $t \leq k$ clusters $X_{j_1}, ..., X_{j_t}$. 
For any fixed choice of $t$ clusters $X_{j_1}, ..., X_{j_t}$, our algorithm returns a list of $(\frac{k}{\veps})^{O(\frac{t}{\veps})}$ $t$-center sets such that (whp) at least one of them is ``good" for $X_{j_1}, ..., X_{j_t}$. We will make this notion more precise later in Section~\ref{sec:good-centers}.

\item The sampling algorithm runs in a {\em single} iteration where $\poly(\frac{t}{\veps})$ points from $X$ are $D^2$-sampled w.r.t. $C$. 
We show that good centers for clusters $X_{j_1}, ...., X_{j_t}$ can simultaneously be found from the sampled points and points in the set $C$.\footnote{Note that there is an iteration for probability amplification in algorithm \GC but since the $2^t$ rounds are independent,  they can be executed independently.}
\end{enumerate}
The formal description of the generalised algorithm is given below.
The algorithm below takes as input dataset $X$, an $\alpha$-approximate solution $C$, error parameter $\veps$, and $t$ and outputs a list $\mathcal{L}$ of $t$-center sets. 
We discuss the nice properties of this algorithm next.

\begin{framed}
\GC($X, C, \veps, t$)\\
\hspace*{0.6in} {\bf Inputs}: Dataset $X$, $\alpha$-approximate $C$, accuracy $\veps$, and number of centers $t$\\
\hspace*{0.6in} {\bf Output}: A list $\L$, each element in $\L$ being a $t$-center set\\
\hspace*{0.6in} {\bf Constants}: $\eta = \frac{2^{16} \alpha t}{\veps^4}; \tau = \frac{128}{\veps}$\\
\hspace*{0.2in} (1) \ \ \ $\L \leftarrow \emptyset$\\
\hspace*{0.2in} (2) \ \ \ Repeat $2^t$ times:\\
\hspace*{0.2in} (3)\hspace*{0.3in}  \ \ \ Sample a multi-set $M$ of $\eta t$ points from $X$ using $D^2$-sampling w.r.t. center set $C$\\
\hspace*{0.2in} (4)\hspace*{0.3in}  \ \ \ $M \leftarrow M \cup$ \{$\frac{128 t}{\veps}$ copies of each element in $C$\}\\
\hspace*{0.2in} (5)\hspace*{0.3in} \ \ \ For all disjoint subsets $S_1, ..., S_t$ of $M$ such that $\forall i, |S_i| = \tau$:\\
\hspace*{0.2in} (6)\hspace*{0.9in} $\L \leftarrow \L \cup (\mu(S_1), ..., \mu(S_t))$\\
\hspace*{0.2in} (7) \ \ \ return($\L$)
\end{framed}
Note that the list size produced by the above algorithm is $|\L| = (\frac{k}{\veps})^{O(\frac{t}{\veps})}$ and running time is $O(nd |\L|)$.
We will show that the \GC algorithm behaves well (whp) for {\bf any} fixed set of $t$ clusters $X_{j_1}, ..., X_{j_t}$ out of clusters $X_1, ..., X_k$.
What this means is the following: Let $X_{j_1}, ..., X_{j_t}$ denote any fixed set of $t$  clusters. 
The list $\L$ produced by the \GC algorithm, with high probability, will contain a $t$-center set $\mathcal{C}$ such that 
$$\psi(\mathcal{C}, \{X_{j_1}, ..., X_{j_t}\}) \leq \left(1+ \frac{\veps}{2} \right) \cdot \sum_{i=1}^{t} \Delta(X_{j_i}) + \frac{\veps}{2} \cdot OPT.$$ 
Note that $OPT$ denotes the optimal cost with respect to the clustering $\X = \{X_{1}, ..., X_{k}\}$. That is $OPT = \Delta(\X) = \sum_{i=1}^{k} \Delta(X_i)$.
We formally state our result as the next theorem.

\begin{theorem}\label{thm:2}
Let $0 < \veps \leq \frac{1}{2}$ and $t$ be any positive integer. 
Let $X_{j_1}, ..., X_{j_t}$ denote an arbitrary set of $t$ clusters out of $k$ clusters $X_1, ..., X_k$ of the dataset $X$.
Let $\L$ denote the list returned by the algorithm \GC($X, C, \veps, t$). 
Then with probability at least $\frac{3}{4}$, $\L$ contains a center set $\mathcal{C}$ such that:
\[
\psi \left(\mathcal{C}, \{X_{j_1}, ..., X_{j_t}\} \right) \leq \left(1+ \frac{\veps}{2} \right) \cdot \sum_{i=1}^{t} \Delta(X_{j_i}) + \frac{\veps}{2} \cdot OPT \leq (1+\veps) \cdot OPT,
\]
where $OPT = \sum_{i=1}^{k} \Delta(X_i)$.
\end{theorem}
We shall formally prove the above theorem in Section~\ref{sec:good-centers}. 
In order to discuss the applications of the \GC algorithm, let us note some of its interesting properties.
Note that the algorithm is essentially a single iteration algorithm. 
The outer loop of size $2^t$ consists of  independent iterations and can be executed independently. 
The rest of the algorithm clearly follows a single line of control and does not have dependencies.
This allows us to design (i) constant pass streaming algorithms (using {\em reservoir sampling}) and (ii) parallel algorithms.
The second useful property is that it finds a good list for any fixed set of $t \leq k$ clusters (whp).
This allows us to exploit the algorithm in certain contexts where once good centers for a few ``bad" clusters have been chosen, choosing good centers for the remaining clusters becomes easy.
We discuss the applications of our algorithm in the subsequent subsections.

An interesting point to note about the \GC algorithm is that the $k$-center set $C$ that it takes as input is only a constant factor approximate solution for the $k$-means problem (i.e., unconstrained version) and not any constrained version. 
Note that we will use the algorithm for designing PTAS for various constrained versions but constant factor solutions for those are not required. 
So in some sense, the \GC algorithm can be seen as an effective way of converting a constant factor approximate solution for the $k$-means problem to PTAS for various constrained versions.
Let us now discuss the applications of our algorithm.

\subsection{Streaming algorithms}
We saw in the previous discussion how an algorithm for the list-$k$-means problem can be converted to a PTAS for a constrained $k$-means problem given that there is a {\em partition algorithm} that finds a feasible clustering with the smallest $k$-means cost.
Examining the \GC algorithm closely, we realise that it can be implemented in $2$-passes using small amount of space. This opens the door for designing streaming PTAS for the constrained versions of the $k$-means problem. If one can design a streaming version of the partition algorithm for some constrained $k$-means problem, then combining it with the streaming version of the \GC algorithm will give us a streaming PTAS for the problem. So, let us first discuss how a streaming version of the \GC algorithm can be designed.

The first bottleneck in designing a streaming version of \GC is that we need a constant factor approximate solution $C$ for the $k$-means problem (i.e., the unconstrained $k$-means problem). Fortunately, there exists a $1$-pass, logspace streaming algorithm that gives a constant factor approximate solution to the $k$-means problem~\cite{brav11}.
Given $C$, we need to show how to implement step~(3) of the algorithm in a streaming manner (the $2^t$ repetitions can be performed independently, this appears as a multiplicative factor in the space used). The probability of sampling a point $p$ is proportional to $\Phi(C, p)$, with the constant of proportionality being $\Phi(C, X)$. 
The sampling can be performed using the ideas of ``reservoir sampling'' (see e.g.~\cite{vitter85}).\footnote{\underline{\it Reservoir sampling}: Given a stream of $n$ data items with associated weights $w_1, ..., w_n$, reservoir sampling stores a single item while making a pass over the data. The $i^{th}$ data item replaces the stored item with probability $\frac{w_i}{\sum_{j=1}^{i}w_j}$. Simple telescoping product shows that the stored item has the same distribution as an item sampled from the distribution $\{\bar{w}_1, ..., \bar{w}_n\}$, where $\bar{w}_i = \frac{w_i}{\sum_{j=1}^{n}w_j}$.}
Since we need to sample $\eta t \leq \poly(\frac{k}{\veps})$ points in step~(3), reservoir sampling takes $O \left( \poly(\frac{k}{\veps}) \cdot \log n \right)$ space. 
Given a sample $M$, steps (5)-(6) can be implemented in $O(|M|^{k \tau})$ space, where $\tau = O(\frac{1}{\veps})$. This can be summarised formally as the following useful lemma that we will prove in Section~\ref{sec:streaming} (we assume that storing a point accounts for one unit of space).

\begin{lemma}
The algorithm \GC can be implemented using $2$-passes over the input data while maintaining space of $O(f(k, \veps) \cdot \log{n})$, where $f(k, \eps) = \left( \frac{k}{\veps} \right)^{O(\frac{k}{\veps})}.$
\end{lemma}
Let us now see how to design a streaming PTAS for a constrained $k$-means problem using the above lemma.
Let $\mathcal{P}^{\C}$ denote the partition algorithm for this constrained problem and suppose there is a streaming version $\mathcal{SP}^{\C}$ of this partition algorithm. 
We will use the $2$-pass streaming version of the \GC algorithm to output the list $\L$. We will then use $\mathcal{SP}^{\C}$ on each element of $\L$ (independently) and  pick the best solution. Since $|\L|$ is small, so is the space requirement. 
From the previous discussion, we know that (whp) we are guaranteed to obtain a $(1+\veps)$-approximate solution. Hence we get a constant pass streaming PTAS. 
So, as long as there is a streaming partition algorithm for a constrained $k$-means problem, there is also a streaming PTAS. 
Now the question is whether there are constrained $k$-means problems for which such streaming partition algorithms can be designed?
Interestingly, we can design such streaming partition algorithms for four out of the six constrained $k$-means problems in Table~\ref{table:1}.
Such streaming PTAS exist~\cite{FrahlingSohler05,FeldmanMS07} for the $k$-means problem (i.e., the classical, unconstrained version) based on the notion of {\em coreset}. 
The following is our  main result on streaming algorithms, details appear in Section~\ref{sec:streaming}. Here, $\Delta$ is the {\em aspect ratio}, i.e., $\Delta = \frac{\max_{p \in X, c \in C}||p-c||}{\min_{p \in X\setminus C, c \in C}||p-c||}$.

\begin{theorem}
There is a $(1 + \veps)$-approximate, $4$-pass, streaming algorithm for the following constrained $k$-means clustering problems that uses $O(f(k, \veps)\cdot (\log{\Delta} + \log n))$-space and $O(d \cdot f(k, \veps))$ time per item, where $f(k, \veps) = (\frac{k}{\veps})^{O(\frac{k}{\veps})}$:
\begin{enumerate}
\item $k$-means clustering\footnote{The classical $k$-means problem can also be seen as a constrained $k$-means problem where there are no constraints.}
\item $r$-gather $k$-means clustering
\item $r$-capacity $k$-means clustering
\item Fault tolerant $k$-means clustering
\item Semi-supervised $k$-means clustering
\end{enumerate}
Further, the space requirement can be improved to $O(f(k, \veps)\cdot \log n)$ using a 5-pass streaming algorithm. 
\end{theorem}
Note that two constrained versions of constrained $k$-means problems from Table~\ref{table:1} are missing from the theorem above. These are the {\em chromatic $k$-means clustering} and the {\em $l$-diversity clustering} problem. In Section~\ref{sec:streaming}, we will argue that deterministic logspace streaming algorithms for these problems are not possible.

\paragraph*{Comparison with Coreset based streaming algorithms}
Streaming {\em coreset} constructions provide another approach to designing streaming algorithm for the $k$-means problem. 
An $(\veps, k)$ coreset of a dataset $X \subset \R^d$ is a weighted set $S \subset \R^d$ along with a weight function $w:S \rightarrow \R^{+}$ such that for any $k$-center-set $C$, we have:
$$
\lvert \sum_{s \in S} \min_{c \in C}w(s) \cdot \norm{s - c}^2 - \sum_{x \in X} \min_{c \in C}\norm{x - c}^2\rvert \leq \veps \cdot \sum_{x \in X} \min_{c \in C}\norm{x - c}^2.
$$
So, it is sufficient to find good $k$-center-set for a coreset $S$ (instead of the dataset $X$). 
There exists one-pass streaming coreset construction~\cite{FeldmanMS07} that uses $poly(k, \frac{1}{\veps}, \log{n})$ space and outputs a coreset of size $poly(k, \frac{1}{\veps}, \log{n})$.
Using this, one can design a single-pass streaming algorithm for the $k$-means problem by first running the streaming algorithm to output a coreset and then finding a good $k$ center set for the small coreset. If the output is supposed to be a clustering, then we will need to make another pass over the data. 
Note that the same idea of working on coreset does not trivially carry over to the constrained versions of $k$-means as there are additional constraints. 
However, there is a specific geometric coreset construction which works for constrained versions of $k$-means. This is one of the first coreset constructions for $k$-means by Har-Peled and Mazumdar~\cite{Har-PeledM04} where the points in the coreset are such that the sum total of the distance of the data points to the nearest coreset point is small. The weight of a coreset point is simply the number of data points for which the coreset point is the closest. So, a coreset point {\em represents} a subset of data points.
Schmidt \etal~\cite{sss20} used this construction for a contrained version called Fair $k$-means.
This coreset construction can be performed in a single pass over the data. 
The coreset size is $O(k \veps^{-d} \log{n})$ and it can be computed in as much space using ideas developed later (e.g., ~\cite{fgsss13}). 
Even though this gives a one-pass algorithm for producing a good center set (two passes for producing clustering), the space requirement is exponentially large in the dimension.
Fortunately, in a more recent development by Makarychev \etal~\cite{mmr18} showed that  the $k$-means cost of {\em any} clustering is preserved up to a factor of $(1 + \veps)$ under a projection onto a random $O \left( \frac{\log{(k/\veps)}}{\veps^2}\right)$-dimensional subspace. 
This result when combined with the geometric coreset construction of Har-Peled and Mazumdar~\cite{Har-PeledM04} gives a one-pass, $O \left( \left( \frac{k}{\veps}\right)^{\frac{1}{\veps^2}} \cdot \log{n}\right)$-space algorithm for producing a good $k$-center-set for {\em any} constrained version of the $k$-means problem.
Even though the space bound has a slightly worse dependency on $1/\veps$ than our list-$k$-means based idea, the dependency on $k$ and number of passes is much better.
Indeed, we overlooked this connection with coreset of Har-Peled and Mazumdar and dimension reduction of Makarychev \etal in the previous version when we were designing our list-$k$-means based streaming algorithms and were made to realise this at a later stage of this work.
At this point, all we can say is that designing streaming algorithm based on list-$k$-means is another way of approaching constrained $k$-means problem.

\subsection{Algorithm under stability/separation}
The worst-case complexity of the $k$-means problem is well understood. 
As discussed earlier, the problem is $\mathsf{NP}$-hard and $\mathsf{APX}$-hard.
Hence, various {\em beyond worst-case} type results have been explored in the context of the $k$-means problem and one such direction is clustering under some ``clusterability" condition. 
That is, design algorithms for datasets that satisfy some mathematical condition that captures the fact that the data is ``clusterable" or in other words the data has some meaningful clusters.
Clusterability is captured in various ways using notions such as ``separability" and ``stability".
Separability means that the target clusters are separated in some geometrical sense and stability means that the target clustering does not change under small perturbations of the input points. 
Separability and stability are closely related in various contexts where one implies the other.
A lot of work has been done the area of algorithm design for the $k$-means problem under various clusterability conditions.

The early notions of separation conditions were based on the cost of the optimal $k$-means solution. 
These were defined in the works of Ostrovsky \etal~\cite{orss} and Kumar \etal~\cite{kss}.
The main idea here is to study the behaviour of the optimal $k$-means cost $\Gamma_{X}(k)$ as a function of $k$. 
Clearly, $\Gamma_X(.)$ is a decreasing function for any dataset $X$ since the optimal $k$-means cost cannot increase as $k$ increases.
If the value of $\Gamma_X(i)$ is significantly smaller than $\Gamma_X(i-1)$, then it makes sense to cluster into $i$ clusters than $(i-1)$ clusters.
This idea can be used to estimate the ``right" value of $k$, the number of clusters, in many practical scenarios where the number of clusters cannot be a-priori determined.
This separation condition is commonly referred to by the name {\em ORSS property} (based on the authors of the paper~\cite{orss}) and the irreducibility property~\cite{kss}.
This is formally defined below.

\begin{definition}[$(1+\gamma)$-irreducibility]
Let $\gamma > 0$.
A $k$-means instance $(X, k)$ is said to be $(1+\gamma)$-irreducible if $OPT_{k-1} \geq (1+\gamma) \cdot OPT_k$, where $OPT_i$ denotes the optimal $i$-means cost for the dataset $X$.
\end{definition}
Subsequently, a number of such cost-based separation notions were defined and algorithms under such notions were given. 
This includes {\em weak deletion stability} by Awasthi \etal~\cite{abs10}, {\em approximation stability} by Balcan \etal~\cite{bbg}, and $\beta$-distributed property by Awasthi \etal~\cite{abs10} defined next.

\begin{definition}[$\beta$-distributed]
A $k$-means instance $(X, k)$ is called $\beta$-distributed iff the following holds for any optimal clustering $\{X^{\star}_1, ..., X^{\star}_k\}$:
\[
\forall i, \forall x \notin X^{\star}_i, ||x - \mu(X^{\star}_i)||^2 \geq \beta \cdot \frac{\OPT}{|X_i^{\star}|}.
\]
\end{definition}
We will discuss these stability properties and their relationship in detail in Section~\ref{sec:fpt-as}.
It can be argued that the $\beta$-distributed property of Awasthi \etal~\cite{abs10} is the weakest separation property among the ones mentioned above.
Hence, any result for datasets satisfying the $\beta$-distributed condition will have consequences for datasets satisfying stronger conditions.
So the question is: are there good algorithms for datasets under this condition? 

Awasthi \etal~\cite{abs10} gave a {PTAS} for the $k$-means/median problems on datasets that satisfy the $\beta$-distributed assumption. 
The running time has polynomial dependence on the input parameters $n, k, d$ and exponential dependence on $\frac{1}{\beta}$ and $\frac{1}{\veps}$ ($\veps$ is the accuracy parameter). 
Even though they showed that the super-polynomial dependence on $\frac{1}{\beta}$ and $\frac{1}{\veps}$ cannot be improved,  improving the dependence on other input parameters was left as an open problem. 
In this work, we address this open problem by giving a faster {PTAS} for the $k$-means problem under the $\beta$-distributed notion. 
The running time of the algorithm for the $k$-means problem by Awasthi \etal~\cite{abs10} is $O(dn^3) (k \log{n})^{\poly(\frac{1}{\beta}, \frac{1}{\veps})}$.
We improve the running time to $O \left(dn^3  \left( \frac{k}{\veps} \right)^{O(\frac{1}{\beta \veps^2})} \right)$.
Note that due to our improvement in running time, our algorithm is also a Fixed Parameter Tractable Approximation Scheme (FPT-AS) for the problem with parameters $k$ and $\beta$. Moreover, the running time does not have an exponential dependence on $k$ that is typically the case for such FPT approximation schemes for general datasets.
We formally state our result as the following theorem. 
We shall discuss the proof of this theorem in Section~\ref{sec:fpt-as}.

\begin{theorem}\label{thm:fpt-as-main}
Let $\veps, \beta > 0$, $k$ be a positive integer, and let $X \subset \R^d$ be a $\beta$-distributed dataset. 
There is an algorithm that takes as input $(X, k, \veps, \beta)$ and outputs a $k$-center set $C$ such that $\Phi(C, X) \leq (1 + \veps) \cdot \OPT$ and the algorithm runs in time $O \left(dn^3  \left( \frac{k}{\veps} \right)^{O(\frac{1}{\beta \veps^2})} \right)$.
\end{theorem}
Our running time improvements over the algorithm of Awasthi \etal~\cite{abs10} comes from using a faster algorithm to find good centers for a few optimal clusters called ``expensive clusters" in the terminology used by Awasthi \etal in their analysis.
They had pointed out that if there were a faster algorithm for finding good centers for a constant number of clusters that they call ``expensive clusters", then the overall running time of their algorithm could be significantly improved.
This is precisely what our \GC algorithm allows us to do. 
The \GC algorithm creates a list such that at least one of the elements of the list is a set of good centers for the expensive clusters.
So, one can execute the algorithm of Awathi \etal repeatedly for every element of the list and then pick the best solution.
The details of this are given in Section~\ref{sec:fpt-as}.

\subsection{Parallel algorithms}
We give a massively parallel {PTAS} for the classical and constrained $k$-means problems. 
This actually just comes from a close inspection of the algorithm \GC (note that for a PTAS we will use $t=k$). 
One quickly realises that most of the steps in the algorithm can be performed independently and hence the algorithm can easily be converted to a massively parallel {PTAS} for the $k$-means problem. 
The main results is given in the theorem below. 
The details of the proof are discussed in Section~\ref{sec:parallel}.

\begin{theorem}\label{thm:parallel-main}
Let $\veps > 0$, $(X, k, d)$ be a $k$-means instance, and let $C$ denote a constant $\alpha$-approximate solution of the $k$-means instance. 
Then there is a parallel algorithm in the shared memory CREW model that takes as input the $k$-means instance, $C$, and $\veps$ and outputs a $(1+\veps)$-approximate solution in parallel time $O \left( \left\lceil \frac{nd 2^{\tilde{O}(k/\veps)}}{N}\right\rceil + \frac{k}{\veps} \log{\frac{k}{\veps}}+\log{(nkd)} \right)$ with $N$ processors.
There is similar parallel algorithm in the CRCW model with running time $O \left( \left\lceil \frac{nd 2^{\tilde{O}(k/\veps)}}{N}\right\rceil + \frac{1}{\veps} +\log{(nkd)} \right)$.
For any constrained version of the $k$-means problem with partition algorithm $\mathcal{P}^{\C}$, there is a parallel algorithm with running time $O \left( \left\lceil \frac{nd 2^{\tilde{O}(k/\veps)}}{N}\right\rceil \cdot  t(n,k,d)+\frac{k}{\veps} \log{\frac{k}{\veps}} + \log{(nkd)} \right)$ in the CREW model with $N$ processors. Here, $t(.)$ denotes the running time of the partition algorithm $\mathcal{P}^{\C}$.
\end{theorem}

\section{Preliminaries}
The $k$-means problem is defined as: given a point set $X \subseteq \R^d$ find a set $C \subset \R^d$ of $k$ points (called centers) such that the following cost function is minimised:
\[
\Phi(C, X) \equiv \sum_{x \in X} \min_{c \in C} ||x - c||^2.
\]
We will use the above cost function repeatedly in our discussion. 
Hence for simplicity, when $C = \{c\}$ is a singleton set, then we use $\Phi(c, X)$ instead of $\Phi(\{c\}, X)$.
The $1$-means problem for any dataset $X \subseteq \R^d$ has the following closed form solution: the point that minimizes the sum of squared Euclidean distances for a dataset $X \subset \R^d$ is the geometric mean (or centroid) $\mu(X)  \equiv \frac{\sum_{x \in X}}{|X|}$.
This follows from the following well-known fact.

\begin{fact}\label{fact:1}
For any $X \subset \R^d$ and $c \in \R^d$, we have $\sum_{x \in X} ||x-c||^2  = \Phi(\mu(X), X) + |X| \cdot ||\mu(X) - c||^2$.
\end{fact}

Our algorithms are based on simple sampling ideas. 
The following sampling result from Inaba \etal~\cite{inaba} will be used in our analysis.
The lemma says that the centroid of a small set of uniformly sampled points from the dataset $X$ is a good center with respect to the $1$-means cost for dataset $X$.

\begin{lemma}[\cite{inaba}]\label{lemma:inaba}
Let $S$ be a set of points obtained by independently sampling $M$ points with replacement uniformly at random from a point set $X \subset \R^d$. Then for any $\delta > 0$,
\[
\pr \left[\Phi(\mu(X), X) \leq \left( 1 + \frac{1}{\delta M} \right) \cdot \Delta(X) \right] \geq (1 - \delta).
\]
\end{lemma}

The main sampling technique that we will use in all our algorithms is called $D^2$-sampling that is also known as {\em importance sampling}.

\begin{definition}[$D^2$-sampling]
Given a set of points $X \subset \R^d$ and another non-empty set of points $C \subset \R^d$, $D^2$-sampling from $X$ w.r.t. $C$ samples a point $x \in X$ with probability $\frac{\Phi(C, \{x\})}{\Phi(C, X)}$. When $C$ is empty, then $D^2$-sampling from $X$ w.r.t. $C$ is just uniform sampling from $X$.
\end{definition}

Our basic set of tools for analysis is very small and simple. 
This means that our analysis can be easily generalised for distance measures other than the Euclidean distance.
One of the properties we use in our analysis is an approximate version of the triangle inequality. 
This is stated as the following simple fact for the Euclidean distance.

\begin{fact}[Approximate triangle inequality]\label{fact:2}
For any $x, y, z \in \R^d$, we have $||x - z||^2 \leq 2 \cdot ||x - y||^2 + 2 \cdot ||y - z||^2$.
\end{fact}

As stated earlier, the optimal $1$-means cost of any dataset $X \subset \R^d$ is denoted by $\Delta(X) = \Phi(\mu(X), X)$.
The {\em Voronoi partitioning} of any dataset $X \subset \R^d$ with respect to center set $\{c_1, ..., c_k\}$ is a partition of $X$ into $X_1, ..., X_k$ such that $\forall i, X_i = \{x \in X | \arg\min_{c \in \{c_1, ..., c_k\}}||x - c|| = c_i\}$. 
Fact~\ref{fact:1} tells us that the Voronoi partitioning $X_1, ..., X_k$ of any dataset $X \subset \R^d$ with respect to any optimal $k$-means solution $\{c_1, ..., c_k\}$ satisfies $\mu(X_i) = c_i$.
So, an alternate way of specifying the output is any optimal Voronoi partitioning of the dataset.
Let $\X^{\star} = \{\Xst_1, ..., \Xst_k\}$ denote an optimal $k$-means clustering of a given dataset $X \subset \R^d$.
We use $\OPT$ denote the optimal cost for the $k$-means problem on the dataset $X$. 
Using previous definitions, we have $\OPT = \Delta(\X^{\star}) \equiv \sum_{i=1}^{k} \Delta(\Xst_i)$.
The notion of approximate solution for the $k$-means problem is well known. 
We may also use the notion of {\em bi-criteria approximation} where a center set $C$ is said to be an $(\alpha, \delta)$-approximate solution if $\Phi(C, X) \leq \alpha \cdot \OPT$ and $|C| \leq \delta k$.

In this work, we are not only interested in the optimal solution for the classical $k$-means problem but also for various constrained versions of the $k$-means problem. 
In the constrained versions of the $k$-means problem, an optimal clustering $\X = \{X_1, ..., X_k\}$ should satisfy certain constraints (such as $\forall i, |X_i| \geq r$) in addition to minimising the  $k$-means cost $\sum_{i=1}^{k}\Delta(X_i)$.
So in general, the constrained $k$-means problem is specified by a set of points $X \subset \R^d$, a positive integer $k$, and a set $\mathbb{C}$ of {\em feasible} clusterings of $X$.
We are given an algorithm $\mathcal{P}^{\mathbb{C}}$ which when given a set of $k$ centers $C = \{c_1, ..., c_k\}$, outputs a feasible clustering $\X = \{X_1, ..., X_k\}$ (i.e., a clustering in $\mathbb{C}$) with the least value of the cost function $\sum_{i=1}^{k} \sum_{x \in X_i} ||x - c_i||^2$.
This algorithm is called the {\em partition algorithm} with respect to the specific version of the constrained $k$-means problem. 
Ding and Xu~\cite{dx15} give such partitioning algorithms for various constrained versions of the $k$-means problem. 
We will use $OPT$ to denote the optimal constrained $k$-means cost. 
So, if $\X = \{X_1, ..., X_k\}$ are the optimal constrained $k$-means clusters, then $OPT = \Delta(\X) = \sum_{i=1}^{k} \Delta(X_i)$.
Note that for any constrained $k$-means problem on a dataset $X \subset \R^d$, $OPT$ is lower bounded by $\OPT$ (the optimal unconstrained $k$-means cost). 
As discussed earlier, we approach the constrained versions of the $k$-means problems through the {\em list $k$-means} problem defined below:
\begin{quote}
{\bf List-$k$-means}: Let $X \subset \R^d$ be the dataset and let $\X = \{X_1, ..., X_k\}$ be an arbitrary clustering of dataset $X$. Given $X$, positive integer $k$, and error parameter $\veps > 0$, find a {\em list} of $k$-center sets such that (whp) at least one of the sets gives $(1+\veps)$-approximation with respect to the cost function: 
\[
\psi(\{c_1, ..., c_k\}, \X) \stackrel{def.}{=} \min_{\textrm{permutation } \pi} \left[ \sum_{i=1}^{k} \Phi(c_{\pi(i)}, X_i) \right].
\]
\end{quote}

\section{Algorithm for the list-$k$-means problem}\label{sec:good-centers}
We discuss the algorithm \GC and its analysis this section.
Let $X \subset \R^d$ be the given dataset and let $\X = \{X_1, ..., X_k\}$ be the unknown clustering and our goal is to find (approximately) good centers for these clusters. 
As discussed earlier, since $\X$ is not given we cannot hope to output a single such $k$-center set. 
We are allowed to output a list of such center sets.
The center set that minimises the cost function $\psi(C, \X)$ is $C = \{\mu(X_1), ..., \mu(X_k)\}$ and we denote this minimum cost by $OPT = \sum_{i=1}^{k} \Phi(\mu(X_i), X_i)$.
Let $X_{j_1}, ..., X_{j_t}$ be some fixed set of $t$ clusters out of clusters $X_1, ..., X_k$.
The rest of the discussion will be with respect to these $t$ clusters.
We restate the algorithm \GC below for ease of exposition.
The algorithm takes as input dataset $X$, an $\alpha$-approximate solution $C$, error parameter $\veps$, and $t$ and outputs a list $\mathcal{L}$ of $t$-center sets.

\begin{framed}
\GC($X, C, \veps, t$)\\
\hspace*{0.6in} {\bf Inputs}: Dataset $X$, $(\alpha, \beta)$-approximate $C$, accuracy $\veps$, and number of centers $t$\\
\hspace*{0.6in} {\bf Output}: A list $\L$, each element in $\L$ being a $t$-center set\\
\hspace*{0.6in} {\bf Constants}: $\eta = \frac{2^{16} \alpha t}{\veps^4}; \tau = \frac{128}{\veps}$\\
\hspace*{0.2in} (1) \ \ \ $\L \leftarrow \emptyset$\\
\hspace*{0.2in} (2) \ \ \ Repeat $2^t$ times:\\
\hspace*{0.2in} (3)\hspace*{0.3in}  \ \ \ Sample a multi-set $M$ of $\eta t$ points from $X$ using $D^2$-sampling w.r.t. center set $C$\\
\hspace*{0.2in} (4)\hspace*{0.3in}  \ \ \ $M \leftarrow M \cup$ \{$\frac{128 t}{\veps}$ copies of each element in $C$\}\\
\hspace*{0.2in} (5)\hspace*{0.3in} \ \ \ For all disjoint subsets $S_1, ..., S_t$ of $M$ such that $\forall i, |S_i| = \tau$:\\
\hspace*{0.2in} (6)\hspace*{0.9in} $\L \leftarrow \L \cup (\mu(S_1), ..., \mu(S_t))$\\
\hspace*{0.2in} (7) \ \ \ return($\L$)
\end{framed}

We will show that the \GC algorithm behaves well (w.h.p.) for {\bf any} fixed set of $t$ clusters $X_{j_1}, ..., X_{j_t}$.
What this means is the following: Let $X_{j_1}, ..., X_{j_t}$ denote any fixed set of $t$ optimal clusters. 
The list $\L$ produced by the \GC algorithm, with high probability, will contain a $t$-center set $\mathcal{C}$ such that 
$$\psi(\mathcal{C}, \{X_{j_1}, ..., X_{j_t}\}) \leq \left(1+ \frac{\veps}{2} \right) \cdot \sum_{i=1}^{t} \Delta(X_{j_i}) + \frac{\veps}{2} \cdot OPT.$$ 
Note that $OPT$ denotes the optimal cost with respect to the clustering $\X = \{X_{1}, ..., X_{k}\}$. That is $OPT = \Delta(\X) = \sum_{i=1}^{k} \Delta(X_i)$.
We formally state our result as the next theorem. 
Note that this is the restatement of Theorem~\ref{thm:2} in the Introduction.

\begin{theorem}
Let $0 < \veps \leq 1/2$ and $t$ be any positive integer. 
Let $X_{j_1}, ..., X_{j_t}$ denote an arbitrary set of $t$ clusters.
Let $\L$ denote the list returned by the algorithm \GC($X, C, \veps, t$). 
Then with probability at least $3/4$, $\L$ contains a center set $\mathcal{C}$ such that
\[
\psi \left(\mathcal{C}, \{X_{j_1}, ..., X_{j_t}\} \right) \leq \left(1+ \frac{\veps}{2} \right) \cdot \sum_{j=1}^{t} \Delta(X_{j_i}) + \frac{\veps}{2} \cdot OPT,
\]
where $OPT = \sum_{i=1}^{k} \Delta(X_i)$.
\end{theorem}

WLOG, we will assume that $j_i = i$, that is the $t$ clusters $X_{j_1}, ..., X_{j_t}$ are the first $t$ clusters $X_1, ..., X_t$. 
Since, the input center set $C$ is an $(\alpha, \beta)$-approximate solution to the classical $k$-means problem on dataset $X$, we have
\begin{equation}\label{eqn:cost}
\Phi(C, X) \leq \alpha \cdot \OPT \quad \textrm{ and } \quad |C| \leq \beta k
\end{equation}
Note that the outer iteration (repeat $2^t$ times in line (2)) is to amplify the probability that the list $\L$ containing a good $t$-center set. We will show that the probability of finding a good $t$-center set in one iteration is at least $(3/4)^t$ and the theorem follows from simple probability calculation. So in the remaining discussion we will only discuss one iteration of the algorithm.
Consider the multi-set $M$ after line (3) of the algorithm. 
We will show that with probability at least $(3/4)^t$, there are disjoint (multi) subsets $T_1, ..., T_t$ each of size $\tau$ such that for every $j = 1, ..., t$, 
\begin{equation}\label{eqn:reqd}
\Phi(\mu(T_j), X_j) \leq \left(1 + \frac{\veps}{2} \right) \cdot \Delta(X_j) + \frac{\veps}{2t} \cdot OPT.
\end{equation}
Since we try out all possible subsets in step (5), we will get the desired result.
More precisely, we will argue in the following manner:
consider the multi-set  $C' = \{\frac{16t}{\veps} \textrm{ copies of each element in $C$}\}$. 
We can interpret $C'$ as a union of multi-sets $C_1', C_2', ..., C_t'$, where $C_j' = \{\frac{16}{\veps} \textrm{ copies of each element in $C$}\}$.
Also, since $M$ consists of $\eta t$ independently sampled points, we can interpret $M$ as a union of multi-sets $M_1', M_2', ..., M_t'$ where $M_1'$ is the first $\eta$ points sampled, $M_2'$ is the second $\eta$ points and so on.
For all $j = 1, ..., t$, let $M_j = C_j' \cup (M_j' \cap X_j)$.\footnote{$M_j' \cap X_j$ in this case, denotes those points in the multi-set $M_j'$ that belongs to $X_j$.}
We will show that for every $j \in \{1, ..., t\}$, with probability at least $(3/4)$, $M_j$ contains a subset $T_j$ of size $\tau$ that satisfies eqn. (\ref{eqn:reqd}). 
Note that $T_j$'s being disjoint follows from the definition of $M_j$.
It will be sufficient to prove the following lemma.

\begin{lemma}\label{lem:list-mainlem}
Consider the sets $M_1, ..., M_t$ as defined above. For any $j \in \{1, ..., t\}$, 
\[
\pr \left[\exists T_j \subseteq M_j \textrm{ s.t. } |T_j| = \tau \textrm{ and } \left(\Phi(\mu(T_j), X_j) \leq \left(1 + \frac{\veps}{2} \right) \cdot \Delta(X_j) + \frac{\veps}{2t} OPT \right)\right] \geq \frac{3}{4}.
\]
\end{lemma}
We prove the above lemma in the remaining discussion. 
We do a case analysis that is based on whether $\frac{\Phi(C, X_j)}{\Phi(C, X)}$ is large or small for a particular $j \in \{1, ..., t\}$.

\noindent
- \underline{\it Case-I} $\left(\Phi(C, X_j) \leq \frac{\veps}{6 \alpha t} \cdot \Phi(C, X) \right)$:
Here we will show that there is a subset $T_j \subseteq C_j' \subseteq M_j$ that satisfies eqn. (\ref{eqn:reqd}).

\noindent
- \underline{\it Case-II} $\left(\Phi(C, X_j) > \frac{\veps}{6 \alpha t} \cdot \Phi(C, X) \right)$:
Here we will show that $M_j$ contains a subset $T_j$ such that $\Phi(\mu(T_j), X_j) \leq \left(1+\frac{\veps}{2} \right) \cdot \Delta(X_j)$ and hence $T_j$ also satisfies eqn. (\ref{eqn:reqd}).

We discuss these two cases next.
The analysis is similar to the analysis of the $D^2$-sampling based algorithm for $k$-means by Bhattacharya \etal~\cite{bjk}. 
Since there are a few crucial differences, and for the sake of clarity we continue with the detailed proof in Appendix~\ref{app:list-mainlem}.

\section{Streaming algorithm for constrained $k$-means}\label{sec:streaming}
In this section, we extend the algorithms for the list-$k$-means problem to the streaming setting. In the streaming model, we are allowed to make constant number of passes over the data. However, the algorithm is allowed to maintain small amount of space. We consider the insertion only stream and measure space in terms of the number of data points being stored. The following result follows easily from prior work.

\begin{lemma}
\label{lem:str}
The algorithm \GC can be implemented using $2$-passes over the input data while maintaining space of $O(f(k, \veps) \cdot \log{n})$, where $f(k, \eps) = \left( \frac{k}{\veps} \right)^{O(\frac{k}{\veps})}.$
\end{lemma}

\begin{proof}
We show how to implement the steps of this procedure in the streaming setting. The algorithm of Braverman et al.~\cite{brav11} gives a single pass constant factor algorithm for the $k$-means while maintaining space of $O(k \cdot \log n)$. We use this procedure to get the initial set $C$ of centers in the \GC procedure. 

Given $C$, we need to show how to implement step~(3) of the procedure in a streaming manner (the $2^t$ repetitions can be performed in parallel, this appears as a multiplicative factor in the space used). The probability of sampling a point $p$ is proportional to $\Phi(C, p)$, with the constant of proportionality being $\Phi(C, X)$. The sampling can be performed using the ideas of ``reservoir sampling'' (see e.g.~\cite{vitter85} and~\cite{bijk18} for a more detailed discussion on the space usage in the streaming setting).\footnote{\underline{\it Reservoir sampling}: Given a stream of $n$ data items with associated weights $w_1, ..., w_n$, reservoir sampling stores a single item while making a pass over the data. The $i^{th}$ data item replaces the stored item with probability $\frac{w_i}{\sum_{j=1}^{i}w_j}$. Simple telescoping product shows that the stored item has the same distribution as an item sampled from the distribution $\{\bar{w}_1, ..., \bar{w}_n\}$, where $\bar{w}_i = \frac{w_i}{\sum_{j=1}^{n}w_j}$.}
Since we need to sample $\eta t \leq \poly(\frac{k}{\veps})$ points in step~(3), reservoir sampling takes $O \left( \poly(\frac{k}{\veps}) \cdot \log n \right)$ space. Given a sample $M$, steps (5)-(6) can be implemented in $O(|M|^{k \tau})$ space, where $\tau = O(\frac{1}{\veps})$.  \qed
\end{proof}

Now we use the above result to give constant pass streaming algorithms for the constrained $k$-means problems. Recall that an instance of the constrained $k$-means problem is specified by a set of valid clusterings (into $k$ disjoint parts) of the input set of points. Our algorithms use the following subroutine (this subroutine is also needed for all known algorithms for constrained $k$-means, see e.g.~\cite{dx15,bjk}): let $C=\{c_1, \ldots, c_k\}$ be a set of 
$k$ points (or ``centers'') and $X$ be the input set of points. Then there is a procedure (called ``partition algorithm'') $\mathcal{P}^{\mathbb{C}}(X, C)$ that outputs a feasible partition of $X$ into $X_1, \ldots, X_k$ such that $ \sum_{i=1}^k \sum_{p \in X_i} ||p-c_i||^2$ is minimized. 
In the streaming setting, it will be useful to break the partition algorithm  into two parts.
Let $\bar{\mathcal{P}}^{\C}(X, C)$ denote a procedure that only outputs the {\em cost} of the optimal clustering. 
Note that this may be simpler than producing the optimal clustering.
Let $D$ denote the data structure created by $\bar{\mathcal{P}}^{\C}(X, C)$ during its execution. 
Then $D$ can be used (instead of $C$) by the partitioning procedure $\mathcal{P}^{\C}$. Consider the following corollary in this framework.

\begin{corollary}
\label{cor:str}
Suppose the procedure $\mathcal{\bar{P}}^{\mathbb{C}}(X, C)$ can be implemented in a $\bar{c}$-pass streaming manner using $\bar{S}$ space and $\mathcal{P}^{\mathbb{C}}(X, D)$ can be implemented in a $c$-pass streaming manner using $S$ space. Then there is a $(2+\bar{c} + c)$-pass $(1+\veps)$-approximate streaming algorithm for the corresponding $k$-means problem. The algorithm uses 
 $f(k, \veps) \cdot \left( \log{n} + \bar{S} + S + |D| \right)$ space, where $f(k, \veps) = \left( \frac{k}{\veps} \right)^{O(\frac{k}{\veps})}.$
\end{corollary}

\noindent
{\bf Remark: } We note that this corollary holds even when the procedure $\mathcal{\bar{P}}^{\mathbb{C}}(X, C)$ outputs $(1+\veps)$ approximate cost and $\mathcal{P}^{\mathbb{C}}(X, D)$ outputs a $(1+\veps)$-approximate clustering of the input points.

\begin{proof}
We first use the algorithm in Lemma~\ref{lem:str} to generate a list $\L$ of $k$-center sets. 
For each such set $C$ in the list $\L$, we run the algorithm 
$\mathcal{\bar{P}}^{\mathbb{C}}(X, C)$ to get the approximate cost of the corresponding optimal clustering. 
Let $C'$ be the $k$-center set that gives the minimum cost and $D'$ denote the data structure created while executing $\mathcal{\bar{P}}^{\mathbb{C}}(X, C')$.
Finally, we use $\mathcal{P}^{\C}(X, D')$ to output the clustering with the least cost.
This is guaranteed to yields a PTAS (see discussions in Introduction and also see~\cite{bjk}). \qed
\end{proof}

\subsection{Partition Algorithms}

We now consider several variants of constrained $k$-means problem and give one pass streaming algorithms for the corresponding partition problems.

\paragraph{\bf {Classical $k$-means}:} Here $\mathcal{\bar{P}}^{\mathbb{C}}(X, C)$ makes one pass over the input $X$. When it sees a point $p$, it assigns it to the closest center $C$, and maintains the cost of this assignment. It outputs the total assignment cost. This runs in a single pass. $\mathcal{P}^{\mathbb{C}}$ is the same as $\mathcal{\bar{P}}^{\mathbb{C}}$, except that it outputs the corresponding clustering and hence runs in a single pass. 
From Corollary~\ref{cor:str}, we get that there is a $4$-pass, log-space, $(1+\veps)$-approximation streaming algorithm for the $k$-means problem.

\vspace{0.2in}

For more non-trivial constrained $k$-means problem, the procedure $\mathcal{P}^{\mathbb{C}}(X, C)$ is often a flow formulation on the following bi-partite graph $G$: on one side (say the left side), we have the set of points $X$ and on the right side, we have $C$. For every point $p \in X, c \in C,$ there is an edge $(p,c)$ in the graph with cost $||p-c||^2$. Clearly, the streaming algorithm cannot maintain this bi-partite graph. Instead, it maintains a compressed version of this graph whose size is a polynomial in $k, \veps, \log n.$ This motivates the following definition.

\begin{definition}
Let $G = (V,E)$ be an edge-weighted bipartite graph with $V = L \cup R$ being the partition of $V$ into the two sides. Further, we associate a number $n_v$ with each vertex $v \in L.$ We say that $G$ {\em represents} the pair $(X, C)$, where $X$ is a set of $n$ points and $C$ is a set of $k$ centers, if the following conditions are satisfied:
\begin{itemize}
    \item The set $R = C$. Each point $p \in X$ is mapped to a unique vertex $v$ in $L$ -- call this vertex $\phi(p)$. Further $n_v$ is equal to $|\phi^{-1}(v)|$.
    \item For each point $p \in X$ and center $c \in X$, the weight of the edge $(\phi(p), c)$ in $G$ is within $(1 \pm \veps)$ of $||p-c||^2$. 
\end{itemize}
\end{definition}

\begin{theorem}
\label{thm:compr}
Given a pair $(X,C)$ of $n$ points and $k$ centers respectively, there is a single pass streaming algorithm which builds a bipartite graph $G$ representing this pair. The space used by this algorithm (which includes the size of $G$) is $O \left((k + \log n) \cdot \left( k \cdot 6^k \cdot  \log \Delta + k^k \cdot \log^k (\frac{1}{\veps})  \right) \right),$ where $\Delta$ is the aspect ratio defined as $\Delta = \frac{\max_{p \in X, c \in C}||p-c||}{\min_{p \in X\setminus C, c \in C}||p-c||}$. 
\end{theorem}

\begin{proof}
For each center $c$, we define a set of buckets $B_c$ as follows: the bucket $b(c,i)$ corresponds to the values $[(1+\veps)^i, (1+\veps)^{i+1})$. Let $d_\min$ and $d_\max$ denote the minimum and the maximum between a pair of points in $X\setminus C$ and $C$ respectively (so $\Delta = d_\max/d_\min).$ We define the buckets $b(c,i)$ for $i = \log_{1+\veps} d_\min, \ldots, \log_{1+\veps} d_\max. $ As above, $B_c$ denotes the collection of buckets corresponding to $c$. Clearly, $|B_c| = O(\frac{\log \Delta}{\eps}).$ 

Now we consider the set $B_{c_1} \times B_{c_2} \ldots \times B_{c_k}$ -- an element of this set is called a {\em hyperbucket.} In other words, a hyperbucket is a $k-$tuple $(b(c_1, i_1), \ldots, b(c_k, i_k))$. Each point $p \in X$ can be mapped to a hyperbucket in the natural manner -- define $\phi(p)$ to be the hyperbucket $(b(c_1, i_1), \ldots, b(c_k, i_k))$, where $i_j$ is such that $||p-c_j|| \in [(1+\veps)^{i_j}, (1+\veps)^{i_{j}+1})$ for $j=1, \ldots, k.$

Call a hyper-bucket ${\bf b}$ to be empty if $\phi^{-1}({\bf b})$ is empty. We now count the number of non-empty hyper-buckets. For a center $c$ and index $i$, call the bucket $b(c,i)$ {\bf interesting} if there is another center $c'$ such that 
$[(1+\veps)^i, (1+\veps)^{i+1}) \cap [\veps \cdot ||c-c'||, ||c-c'||/\veps]$
is non-empty. Call a hyperbucket $(b(c_1, i_1), \ldots, b(c_k, i_k))$ interesting if {\em all} the buckets  $b(c_j,i_j)$ in it are interesting. We first count the number of interesting hyperbuckets:
\begin{claim}\label{cl:hyper} 
The number of interesting hyperbuckets is $O \left(k^k \cdot  \log^k (\frac{1}{\veps}) \right)$.
\end{claim}
\begin{proof}
For a fixed center $c$, the number of interesting buckets $b(c,i)$ is $O(k \log (\frac{1}{\veps})).$ Since all the buckets in a hyperbucket needs to be interesting, the result follows.  \qed
\end{proof}
We now count the number of non-interesting hyperbuckets. 
\begin{claim}
\label{cl:hyper1}
Let $b(c,i)$ be a non-interesting bucket for some center $c$ and index $i$. Then the number of non-empty non-interesting hyperbuckets containing $b(c,i)$ is $O(6^k)$. 
\end{claim}
\begin{proof}
Consider such a hyperbucket ${\bf b}$ containing $b(c,i)$. Let $p$ be a point such that $\phi(p)$ is ${\bf b}$. Let $c'$ be a center other than $c$. Two cases arise:
\begin{itemize}
    \item $(1+\veps)^{i+1} \leq \veps ||c-c'||:$ In this case,
    $$ ||p-c'|| \in ||c-c'|| \pm ||c-p|| \in (1 \pm \veps) \cdot ||c-c'||.$$
    Therefore, there are at most 3 choices for the index $i'$ such that
    $b(c',i')$ is one of the coordinates of ${\bf p}$. 
   \item  $||c-c'|| \leq \veps (1 + \veps)^i$ : Here, 
   $$ ||p-c'|| \in ||c-c'|| \pm ||c-p|| \in (1 \pm \veps) \cdot ||c-p||. $$
   Again, there are at most 3 choices for $i'$ as above. 
\end{itemize}
From the above argument, it is clear that the number of non-empty non-interesting hyperbuckets ${\bf b}$ containing $b(c,i)$ is $O(6^k)$. This proves the claim. \qed
\end{proof}

We can  now count the number of non-empty non-interesting hyperbuckets. 
Consider such a bucket ${\bf b}$. There must be a coordinate $b(c,i)$ in it which is non-interesting -- there are $O(k \log \Delta)$ choices for $b(c,i)$. For each such choice, the above claim shows that there are $O(6^k)$ choices for the remaining coordinates of ${\bf b}$. Thus, we see that the total number of non-empty hyperbuckets is $O(k \cdot 6^k \log \Delta +k^k \cdot  \log^k (\frac{1}{\veps})). $

We can describe the streaming algorithm. The bipartite graph will have all the non-empty hyperbuckets on the left side and the $k$ centers $C$ on the right side. The length of an edge between a hyperbucket $(b(c_1, i_1), \ldots, b(c_k, i_k))$ and a center $c_j$ will be the square of  $(1+\veps)^{i_j}.$ Initially, the left side
of the bipartite graph will be empty (because all hyperbuckets are empty). Whenever a new point $p$ is seen, the algorithm computes $\phi(p)$. If this hyperbucket is not present in the left side of the bipartite graph, we add a new vertex corresponding to it. The algorithm can also maintain the cardinality of $\phi^{-1}({\bf b})$ for every non-empty bucket ${\bf b}$ (this will be stored in the variable $n_v$, where $v$ is a vertex in the left side of the graph). The theorem now follows from the fact that the number of edges is equal to  $k$ times the number of non-empty hyperbuckets, and the space required to maintain $n_v$ values is $\log n$ times the number of non-empty hyperbuckets. This proves the Theorem~\ref{thm:compr}. \qed
\end{proof}

Note that the space required depends on $\log \Delta$, whereas we would really like it to depend on $\log n$ instead. 
We discuss this next.

\paragraph{\bf Removing the dependence on $\Delta$:} Let $D$ denote the set of all pair-wise distances between the centers (so $|D| \leq k^2$). The following is the key observation:

\begin{lemma}
\label{lem:dist}
Let $p$ be a point in $X$, and $c$ be the closest center to it. Let $d$ denote $||p-c||$. For any center $c' \in X$, $||p-c'||$ lies in the range $[u/4, 4u]$ for
some $u \in D \cup \{d\}$.
\end{lemma}
\begin{proof}
Assume $||p-c'|| \geq 4d$, otherwise the lemma is already true. Then 
$$||c-c'|| \geq ||p-c'|| - ||p-c|| \geq 3||p-c||. $$
Therefore, 
$$ ||p-c'|| \leq ||c-c'|| + ||p-c|| \leq \frac{4 ||c-c'||}{3}, $$
and 
$$||p-c'|| \geq ||c-c'|| - ||p-c|| \geq \frac{2 ||c-c'||}{3}. $$
\qed
\end{proof}

We fix a solution $\X$ to the partition problem (of course, the algorithm does not know it, but it will help in the algorithm description). 
For a point $x$, let $d_x$ denote the distance from $x$ to the closest center, 
Let $d^\star$ be the maximum over all points $x$ of $d_x$. We do know $d^\star$, but can find it by performing a pass over the data.

Now suppose we guess the maximum $u \in D$ such that any point $p \in X$ 
which is not assigned to a center in $\X$ within distance $4d_x$ is assigned to a center $c'$ such that 
$||p-c'||$ lies in the range $[u/4,4u]$ (the lemma above ensures that such a value $u$ exists). \footnote{In the algorithm implementation, we will run this for all possible values of $u \in D$. This will increase the space complexity by a factor of $k^2$. } If $d^\star$ happens to be larger than $u$, we update $u$ to $d^\star$. We know that the cost of the solution $\X$ is at least $\Omega(u^2).$

Now, we contract all distances in the metric which are smaller than $u/n^2$ -- again we cannot do this directly in the streaming manner, but whenever a point $p$ arrives, we will view all distances to centers which are less than $u/n^2$ to be $0$.
Since the optimal cost is $\Omega(u^2)$, this distance contraction affects the optimal value by at most a factor of $(1+1/n)$. So now, all non-zero distances between a point $p$ and a center $c \in X$ such that $p$ can be potentially assigned to $c$ lie in the range $[u/n^2, 4u]$. Thus $\Delta$ becomes polynomially bounded. 
However one issue remains -- each point $p$ can only be assigned to a center (other than it's nearest center) which is at most $4u$ distance away. We need to incorporate this fact in the graph structure as well. For each point $p$, let $C(p)$ be the set of centers to which it can be assigned (these are the centers which are at most $4u$ distance away, or the center closest to $p$). Note that there are $2^k$ choices for $C(p)$. 

We  modify the construction of $G$ used in Theorem~\ref{thm:compr} as follows. Recall that the left side of $G$ had one vertex for every hyperbucket ${\bf b}$. Now we will have one vertex for every pair $({\bf b}, C'),$ where $C'$ is a subset of $C$. If $\phi(p)$ is the hyperbucket ${\bf b}$, then we assign $p$ to the pair $({\bf b}, C(p)).$ It is easy to check that with this modification, the arguments in the proof of Theorem~\ref{thm:compr} hold. 
The result of this construction in Theorem~\ref{thm:compr} is that the $\log{\Delta}$ can now be replaced with $\log{n}$ but at the cost of multiplying the overall space requirement by a factor of $k^2 \cdot 2^{k}$ and adding one more pass. This is formally stated as the following Theorem.

\begin{theorem}\label{thm:compr1}
Given a pair $(X,C)$ of $n$ points and $k$ centers respectively, there is a single pass streaming algorithm which builds a bipartite graph $G$ representing this pair. 
The space used by this algorithm (which includes the size of $G$) is $O \left(k^2 \cdot 2^k \cdot (k + \log n + \log \Delta) \cdot \left( k \cdot 6^k \cdot  \log{n} + k^k \cdot \log^k (\frac{1}{\veps})  \right) \right)$ Further, the dependence on $\log \Delta$ can be removed by adding one more pass to the algorithm. 
\end{theorem}

We will now use the above theorem to construct streaming algorithms for a variety of constrained $k$-means problems in the subsequent subsections. 
Note that all we need to do is to discuss the streaming versions of the partition algorithms $\mathcal{\bar{P}}^{\C}$ and $\mathcal{P}^{\C}$.

\subsection{$r$-gather/capacity $k$-means clustering}

Given a set of $k$ centers $C$ and the input data $X$, the partition algorithm $\mathcal{P}^{\mathbb{C}}(X, C)$  for the $r$-gather $k$-means clustering problem needs to partition $X$ into $X_1, \ldots, X_k $ (where $X_i$ corresponds to the set of points assigned to the center $c_i$ in $C$) such  that (i) $|X_i| \geq r$ for each $i$, and (ii) $\sum_i \Phi(c_i, X_i)$ is minimized. 

Given the parameters $r$ and $k$, and the sets $X$ and $C$, the above partition problem can be easily solved using a flow formulation. Indeed, we first build a bipartite graph $G$ as follows. On the left side of this graph, we have the point set $X$ and on the right side $C$. The weight of an edge is the square of the distance between the corresponding points. Finally, we set up a flow formulation by adding a source vertex $s$ and a sink $t$. The edge from $s$ to a vertex $p$ on the left side has capacity 1 (and 0 cost). Similarly, an edge from a vertex $c \in C$ on the right side to $t$ has lower bound $r$ and 0 cost. The edges of the bipartite graph have unbounded capacity. Clearly. a min-cost flow of value $|X|$ gives the optimal solution to the partition problem. The above flow formulation was given by Ding and Xu~\cite{dx15}.

Now we show how to get a $(1+\veps)$-approximation in the streaming setting. Instead of constructing the graph above, we construct the bipartite graph as given by Theorem~\ref{thm:compr1}. We can again formulate the flow formulation in an analogous manner -- the only change is that the edge from $s$ to a hyperbucket $v$ has capacity $n_v$. 
It is clear the optimal solution to this flow-formulation computes the cost of a  $(1+\veps)$-approximate solution for the partition problem. This is basically the algorithm $\mathcal{\bar{P}}^{\C}$.
However, we also need to design the algorithm $\mathcal{P}^{C}$ that outputs a clustering.

Once we have found the optimal flow in this bipartite graph, we make one more pass over the input $X$ to figure out the actual partition. We can do this as follows: we can think of each vertex $v$ on the left-side as having $n_v$ ``copies'' and each such copy is assigned to a center in $C$ by the min-cost flow solution. Now we make one more pass over the data. When we see a point $p$ for which $\phi(p)$ is $v$, we pick any one of the copies at $v$ and identify it with $p$. We assign $p$ to the center assigned by the flow. Now we can remove this copy by decreasing $n(v)$ and reducing flow by 1 unit on appropriate edges. It is easy to check that this will output the desired solution. 
The $r$-capacity $k$-means clustering problem is similar, except that in the flow formulation edges from the right nodes to $t$ has a capacity of $r$ rather than lower bound of $r$.
Combining everything, we get the following lemma.

\begin{lemma}
\label{lemma:gatherstr}
Consider the $r$-gather/capacity $k$-means clustering problem.
There is a $1$-pass streaming algorithm $\mathcal{\bar{P}}^{\C}(X, C)$ that gives $(1+\veps)$-approximation to the cost of the optimal feasible clustering with respect to center set $C$ and a concise data structure $D$. 
Furthermore, there is a $1$-pass streaming algorithm $\mathcal{P}^{\C}(X, D)$ that outputs the $(1+\veps)$-approximate feasible optimal clustering with respect to center set $C$.
The space requirement for both the above algorithms and size of data structure $D$ is  $O(f(k, \veps) \cdot (\log{n}+ \log \Delta)),$ where $f(k, \veps)$ is $k^{O(k)} \cdot \log^k (1/\veps).$ Further, the space requirement can be improved to $O(f(k, \veps)\cdot \log{n})$ using a 5-pass streaming algorithm. 
\end{lemma}

Combining the above the lemma with Corollary~\ref{cor:str}, we get the following theorem for the $r$-gather/capacity $k$-means clustering problem.

\begin{theorem}
There is a $(1 + \veps)$-approximate, $4$-pass, streaming algorithm for the $r$-gather/capacity $k$-means clustering problem that uses $O(f(k, \veps)\cdot (\log{n}+\log \Delta))$-space, where $f(k, \veps) = (\frac{k}{\veps})^{O(\frac{k}{\veps})}$. Further the space requirement can be improved to $O(f(k, \veps)\cdot \log n)$-space using a 5-pass streaming algorithm. 
\end{theorem}

\subsection{Fault tolerant $k$-means clustering}
The fault tolerant $k$-means problem is defined as follows: Given dataset $X$ and integers $k$ and $l \leq k$, find a clustering $\X = \{X_1, ..., X_k\}$ such that the squared distance of the points to the $l$ nearest centers from the set $\{\mu(X_1), ..., \mu(X_k)\}$ is minimised. We solve this problem through a reduction to a specialised version of the {\em Chromatic $k$-means problem}. 
In the Chromatic $k$-means problem each data point has an associated colour and the goal is to find a clustering $\X = \{X_1, ..., X_k\}$ with minimum cost $\Delta(\X)$ such that none of the clusters have more than one point of the same colour. Clearly, the problem is well defined if every colour class has at most $k$ points.
Given an instance of the fault tolerant $k$-means problem $(X, k, l)$, we construct an instance of the chromatic $k$-means problem by replicating every point $l$ times and giving each copy the same colour. 
It is simple to show that a $(1+\veps)$-approximate solution to the constructed instance will give $(1+\veps)$-approximate solution to the instance of the fault tolerant $k$-means problem (see \cite{dx15} for details of this reduction).
Note that we said the reduction is to a ``specialised" version of the chromatic $k$-means problem. 
In the batch setting, the order of the points in the dataset $X$ does not matter. 
However, in the streaming setting it does matter if the points of the same colour are ``bunched" together. 
That is, you will see all points of the same colour before seeing points of another colour in the stream. 
This is because, one does not need to spend space to keep track of colours.
Let us call this version of the chromatic $k$-means problem as the {\em sequential chromatic $k$-means problem} just to differentiate this subtle issue.
So, our reduction is actually to the sequential chromatic $k$-means problem.
We will now just construct streaming partition algorithm for the sequential chromatic $k$-means problem.

In the batch setting the partition algorithm for the chromatic $k$-means problem works by constructing and solving the flow networks (given in the previous subsection) separately for individual colours.
In the streaming setting, we can first make a pass over the data to create a compressed graph for all the data points and then at the end consider the flow graphs for each of the colours one after the other.
We can then solve the flow networks for each of colours one after another to get the approximate solution. 
It is important to note that this is possible with limited memory because points of the same colour appear in the stream together. 
This gives the following lemma for streaming partitioning for the fault tolerant $k$-means problem.

\begin{lemma}
\label{lemma:fault}
Consider the fault tolerant $k$-means clustering problem.
There is a $1$-pass streaming algorithm $\mathcal{\bar{P}}^{\C}(X, C)$ that gives $(1+\veps)$-approximation to the cost of the optimal feasible clustering with respect to center set $C$ and a concise data structure $D$. 
Furthermore, there is a $1$-pass streaming algorithm $\mathcal{P}^{\C}(X, D)$ that outputs the $(1+\veps)$-approximate feasible optimal clustering with respect to center set $C$.
The space requirement for both the above algorithms and size of data structure $D$ is  $O(f(k, \veps) \cdot (\log{n}+\log \Delta)),$ where $f(k, \veps)$ is $k^{O(k)} \cdot \log^k (1/\veps).$ Further, the dependence on $\log \Delta$ can be removed using a 2-pass algorithm. 
\end{lemma}

Combining the above the lemma with Corollary~\ref{cor:str}, we get the following theorem for the fault tolerant $k$-means clustering problem.

\begin{theorem}
There is a $(1 + \veps)$-approximate, $4$-pass, streaming algorithm for the fault tolerant $k$-means clustering problem that uses $O(f(k, \veps)\cdot (\log{n}+\log \Delta))$-space, where $f(k, \veps) = (\frac{k}{\veps})^{O(\frac{k}{\veps})} $. Further, the space requirement can be improved to $O(f(k, \veps)\cdot \log{n})$ using a 5-pass streaming algorithm. 
\end{theorem}

\subsection{Semi-supervised $k$-means clustering}
Recall that the semi-supervised $k$-means clustering problem is defined as follows: Given dataset $X$, $\alpha \in [0, 1]$,  positive integer $k$, and a ``target" clustering $\bar{\X} = \{\bar{X}_1, ..., \bar{X_k}\}$, the goal is to find a clustering $\X = \{X_1, ..., X_k\}$ such that the following distance function is minimised: 
\[
cost(\X) = \alpha \cdot \Delta(\X) + (1 - \alpha) \cdot dist(\X, \bar{\X}),
\]
where the distance function $dist$ denotes the set difference distance. 
Let $OPT$ denote the cost of the optimal solution.

Note that the problem definition does not quite fit into the unified framework that that was given earlier in the introduction.
In the unified framework, the goal was to find a clustering $\X$ with minimum $\Delta(\X)$ that satisfies certain constraints $\C$.
In this problem the goal is to minimise the above distance function and there are no separate constraints.
Suppose there is a partition algorithm $\mathcal{P}(X, \{c_1, ..., c_k\})$ that returns a clustering $\X = \{X_1, ..., X_k\}$ with minimum value of $\alpha \cdot \sum_{i} \Phi(c_i, X_i) + (1 - \alpha) \cdot dist(\X, \bar{\X})$.
Does this guarantee that one can use the partition algorithm with the \GC algorithm to output a $(1+\veps)$-approximate solution (whp)? 
Let $\mathcal{S} = \{S_1, ..., S_k\}$ denote an optimal clustering. 
The \GC algorithm guarantees that there is a $k$-center set $\{c_1, ..., c_k\}$ in the output list $\L$ such that:
\[
\sum_{i} \Phi(c_i, S_i) \leq (1+\veps) \cdot \sum_i \Delta(S_i) = (1 + \veps) \cdot \Delta(\mathcal{S})
\]
Suppose $\mathcal{P}(X, \{c_1, ..., c_k\})$ returns $\mathcal{S}' = \{S_1', ..., S_k'\}$.
Then we claim that $cost(\mathcal{S}') \leq (1+\veps) \cdot cost(\mathcal{S}) = (1+\veps) \cdot OPT$.
This simply holds because
\begin{eqnarray*}
cost(\mathcal{S}') 
&=& \alpha \cdot \sum_{i} \Delta(S_i') + (1 - \alpha) \cdot dist(\mathcal{S}', \bar{\X}) \\
&\leq& \alpha \cdot \sum_{i} \Phi(c_i, S_i') + (1 - \alpha) \cdot dist(\mathcal{S}', \bar{\X}) \\
&\leq& \alpha \cdot \sum_{i} \Phi(c_i, S_i) + (1 - \alpha) \cdot dist(\mathcal{S}, \bar{\X})\\
&\leq& \alpha \cdot (1+\veps) \cdot  \Delta(\mathcal{S}) + (1-\alpha) \cdot dist(\mathcal{S}, \bar{\X}) \\ 
&\leq& (1+\veps) \cdot  \left( \alpha \cdot  \Delta(\mathcal{S}) + (1-\alpha) \cdot dist(\mathcal{S}, \bar{\X}) \right) \\
&=& (1+\veps) \cdot OPT.
\end{eqnarray*}

So, we just need to design a streaming partition algorithm $\mathcal{P}(X, C)$ as described above. 
The algorithm in the batch setting is given by Ding and Xu~\cite{dx15} using a minimum cost flow formulation as for the previous problems.
Since the mapping of centers $\{c_1, ..., c_k\}$ to clusters $S_1, ..., S_k$ is not known, we will try all possible $k!$ permutations.
For a fixed permutation $\pi$, the flow network is setup in the following manner: the cost of the edge $(x, c_i)$ is $\alpha \cdot ||x - c_i||^2$ in case $x \in S_{\pi(i)}$ and $\alpha \cdot ||x - c_i||^2 + (1-\alpha)$, otherwise. 
The other details of the construction is similar to that for the $r$-gather/capacity problem. For the streaming version, we will use our compressed graph idea. The space requirement will be the same as that for the $r$-gather/capacity problem, except that here there will be an extra multiplicative factor of $k!$ because of trying out all possible permutations. 
However, this gets absorbed into the $k^{O(k)}$ factor in the space requirement and we get a lemma very similar to that in the previous subsections.

\begin{lemma}
\label{lemma:semi}
Consider the semi-supervised $k$-means clustering problem.
There is a $1$-pass streaming algorithm $\mathcal{\bar{P}}^{\C}(X, C)$ that gives $(1+\veps)$-approximation to the cost of the optimal feasible clustering with respect to center set $C$ and a concise data structure $D$. 
Furthermore, there is a $1$-pass streaming algorithm $\mathcal{P}^{\C}(X, D)$ that outputs the $(1+\veps)$-approximate feasible optimal clustering with respect to center set $C$.
The space requirement for both the above algorithms and size of data structure $D$ is  $O(f(k, \veps) \cdot (\log{n}+\Delta)),$ where $f(k, \veps)$ is $k^{O(k)} \cdot \log^k (1/\veps).$ Further, the space requirement can be improved to $O(f(k, \veps)\cdot \log{n})$ using a 5-pass streaming algorithm. 
\end{lemma}

Finally, the discussion above gives the following formal result for the semi-supervised $k$-means problem.

\begin{theorem}
There is a $(1 + \veps)$-approximate, $4$-pass, streaming algorithm for the semi-supervised $k$-means clustering problem that uses $O(f(k, \veps)\cdot (\log{n}+ \log \Delta))$-space, where $f(k, \veps) = (\frac{k}{\veps})^{O(\frac{k}{\veps})} $.
Further, the space requirement can be improved to $O(f(k, \veps)\cdot \log{n})$ using a 5-pass streaming algorithm. 
\end{theorem}

\subsection{Impossibility results}
In this subsection, we will argue that a deterministic logspace streaming algorithm for the chromatic $k$-means clustering problem and $l$-diversity clustering problem is not possible. 
The argument was conveyed to us by Khanna and Assadi~\cite{personalcom}, we give the details for sake of compeleteness.
We will give an argument for the chromatic $k$-means clustering problem that can be extended to the $l$-diversity clustering problem.
Recall, that in the chromatic $k$-means problem every point in the dataset $X$ has an associated colour and the goal is to find a clustering $\X = \{X_1, ..., X_k\}$ of the dataset $X$ with smallest cost $\Delta(\X)$ such that no cluster has more than one point of the same colour.
For the sake of contradiction, assume that there is a logspace streaming algorithm $\mathcal{A}$ for this problem.
We will obtain a contradiction using the following cleverly designed dataset: 
All the points in the dataset $X$ are co-located (i.e., have the same coordinates). 
There are $m = \frac{n}{2}$ colours and there are precisely two points with the same colour (assume that $n$ is even). The number of clusters $k = 2$. 
We will argue that it is not possible for any streaming algorithm to partition the data points into two clusters making sure that every two points with same colour are in different clusters unless it uses $\Omega(n)$ space (hence contradicting the existence of $\mathcal{A}$).

Assuming the existence of $\mathcal{A}$, the following simple four party communication protocol for producing a feasible clustering with logarithmic communication cost should also exist. 
Each of the four parties $P_1, P_2, P_3$ and $P_4$ receive a disjoint partition $D_1, D_2, D_3, D_4$ of the dataset $X$. 
First $P_1$ decides the clustering of every point in $D_1$ and then communicates message $\mathcal{M}_1$ to $P_2$. 
After this, $P_2$ decides the clustering of points in $D_2$ and sends message $\mathcal{M}_2$ to $P_3$ who then decides the clustering of points in $D_3$ and sends a message $\mathcal{M}_3$ to $P_4$ who finally decides the clustering of points in $D_4$. 
Let the $m$ colours be denoted by the numbers $\{1, ..., m\}$. We will make use the following combinatorial lemma from coding theory in our counting argument.

\begin{lemma}
There exists a set $S$ of $\frac{m}{2}$-sized subsets of $\{1, ..., m\}$ such that (i) $\forall Y \neq Z \in S$, $|(Y\cup Z) \setminus(Y \cap Z)| \geq \frac{m}{6}$, and (ii) $|S| = 2^{\Omega(m)}$. 
\end{lemma}
Note that since all the points are co-located, the only relevant information for each point is the colour of the point. So, the data can be seen as just a (multi) set of colours.
We can also show the following using the above lemma and the fact that each party sends message of logarithmic size. 

\begin{lemma}
There exists three $\frac{m}{2}$-sized sets $Y_1, Z_1, Z_2$ with the following properties: 
\begin{enumerate}
\item $P_1$ send the same message $M_1$ to $P_2$ on both $Y_1$ and $Z_1$,
\item There exists at least one element $i \in Y_1\setminus Z_1$ that is not present in $Z_2$, and 
\item $P_2$ conditioned on receiving message $M_1$, sends the same message $M_2$ to $P_3$ on both $Y_1$ and $Z_2$.
\end{enumerate}
\end{lemma}
Given the above lemma, consider the following three scenarios for the $4$-party protocol:
\begin{enumerate}
\item {\it Scenario 1}: $P_1$ receives $Y_1$, $P_2$ receives $Y_1$, $P_3$ receives $\{\}$, and $P_4$ receives the remaining points.
\item {\it Scenario 2}: $P_1$ receives $Z_1$, $P_2$ receives $Y_1$, $P_3$ receives $\{i\}$, and $P_4$ receives the remaining points.
\item {\it Scenario 3}: $P_1$ receives $Y_1$, $P_2$ receives $Z_2$, $P_3$ receives $\{i\}$, and $P_4$ receives the remaining points.
\end{enumerate}
We will obtain a contradiction from the following sequence of arguments:
\begin{itemize}
\item {\it Scenario 1}: WLOG assume that $P_1$ assigns $i$ from $Y_1$ to cluster $1$. So, $P_2$ must assign $i$ from $Y_1$ to cluster $2$.

\item {\it Scenario 2}: As $P_2$ gets the same message $M_1$ as in scenario $1$, it must assign $i$ from $Y_1$ to cluster $2$. This means the $P_3$ must assign $i$ to cluster $1$.

\item {\it Scenario 3}: $P_1$ assigns $i$ in $Y_1$ to cluster $1$. Now, since $P_2$ sends the same message $M_2$ to $P_3$ as in scenario $2$, $P_3$ must assign $i$ to cluster $1$ (as in scenario $2$). This is a contradiction because two points with same colour cannot be assigned the same cluster.
\end{itemize}

\section{Faster {PTAS} for $\beta$-distributed $k$-means instances}\label{sec:fpt-as}
A lot of work has been done the area of algorithm design for the $k$-means problem under various clusterability conditions. 
In the next subsection, we first have a discussion on a few cost-based clusterability notions and their interrelationship that are relevant to this work

\subsection{Stability/separation conditions}

The early notions of separation conditions ware based on the cost of the optimal $k$-means solution. 
These were defined in the works of Ostrovsky \etal~\cite{orss} and Kumar \etal~\cite{kss}.
The main idea here is to study the behaviour of the optimal $k$-means cost $\Gamma_{X}(k)$ as a function of $k$. 
Clearly, $\Gamma_X$ is a decreasing function for any dataset $X$ since the optimal $k$-means cost will decrease as $k$ increases.
If the value of $\Gamma_X(i)$ is significantly smaller than $\Gamma_X(i-1)$, then it makes sense to cluster into $i$ clusters than $(i-1)$ clusters.
This idea can be used to estimate the {\em right} value of $k$, the number of clusters, in many practical scenarios where the number of clusters cannot be a-priori determined.
This separation condition is commonly referred to by the name ORSS property (based on the authors of the paper~\cite{orss}) and the irreducibility property~\cite{kss}.
This is formally defined below.

\begin{definition}[$(1+\gamma)$-irreducibility]
Let $\gamma > 0$.
A $k$-means instance $(X, k)$ is said to be $(1+\gamma)$-irreducible if $OPT_{k-1} \geq (1+\gamma) \cdot OPT_k$, where $OPT_i$ denotes the optimal $i$-means cost for the dataset $X$.
\end{definition}

Another very similar notion is that of {\em weak deletion stability} defined by Awasthi \etal~\cite{abs10}. Recall 
that in a solution to the $k$-means problem, each point is {\em assigned} to the closest center. 

\begin{definition}[$(1+\gamma)$-weak deletion stability]
Let $\gamma > 0$, and consider an instance $(X,k)$ of the $k$-means problem. 
Let $\{c_1^{\star}, ..., c_k^{\star}\}$ denote an optimal set of  $k$ centers for this instance. 
This instance is said to be $(1+\gamma)$-weakly deletion stable
if  for any $i \neq j$, $OPT^{(i \rightarrow j)} > (1 + \gamma) \cdot \OPT$, where $\OPT$ denotes the optimal $k$-means cost and $OPT^{(i \rightarrow j)}$ denotes the cost of clustering obtained by removing $c_i^{\star}$ as a center and assigning all the
points which were assigned to it to the  center $c_j^{\star}$.
\end{definition}

The next simple lemma establishes that weak deletion stability condition is a weaker condition than the irreducibility condition. 

\begin{lemma}[Claim 3.3 in ~\cite{abs10}]
Any $(1+\gamma)$-irreducible dataset is also $(1+\gamma)$-weakly deletion stable.
\end{lemma}

Let us see why the above conditions can be interpreted as separation conditions. 
Let us discuss in terms of the weak deletion stability condition since this is the weaker condition of the above two. 
The following simple and well-known fact (restatement of Fact~\ref{fact:1} in preliminaries) will be used in the discussion.

\begin{fact}
The following holds for any point set $X \subseteq \R^d$ and any point $p \in \R^d$:
\[
\sum_{x \in X}||x - p||^2 = \sum_{x \in X} ||x - \mu(X)||^2 + |X| \cdot ||p - \mu(X)||^2,
\]
where $\mu(X) = \frac{\sum_{x \in X} x}{|X|}$ is the centroid of the point set.
\end{fact}

Note that the above fact implies that the optimal $1$-means solution for any dataset is the centroid of the dataset. 
Let us fix a few notations that we will use in the remaining discussion.
For any $k$-means instance $(X, k)$, we will use  $X_1^{\star}, ..., X_k^{\star}$ to denote the optimal clusters and $c_1^{\star}, ..., c_k^{\star}$ denote the optimal cluster centers.
As defined earlier, the optimal $k$-means cost is $\OPT$ and the optimal $k$-means cost of the $i^{th}$ cluster $X_i^{\star}$ is $\Delta(X^{\star}_i)$.
That is, $\OPT = \sum_{i=1}^{k} \Delta(X^{\star}_i)$.
Consider any $k$-means instance $(X, k)$ that is $(1+\gamma)$-weak deletion stable and any $i \neq j$. 
Using the above fact, we have $OPT^{i \rightarrow j} = \Delta(X^{\star}_i) + |X_i^{\star}| \cdot ||c_i^{\star} -  c_j^{\star}||^2 > (1 + \gamma) \cdot \OPT$ which further implies that 

\begin{equation}\label{eqn:fptas-1-1}
\forall i \neq j, ||c_i^{\star} - c_j^{\star}||^2 > \frac{(1 + \gamma)\cdot \OPT - \Delta(X^{\star}_i)}{|X_{i}^{\star}|}
\end{equation}
The quantity on the right of eqn. (\ref{eqn:fptas-1-1}) can be quite large indicating large separation between the optimal centers.
This, in turn, indicates large separation between optimal clusters.
Another cost-based property was defined by Balcan \etal~\cite{bbg} that they called the {\em approximation stability} condition\footnote{This is popularly also known as the {\em BBG property} due to the name of the authors.}:

\begin{definition}[$(1+\gamma, \delta)$-approximation stability]
A $k$-means instance $(X, k)$ is said to be $(1+\gamma, \delta)$ approximation stable instance with respect to an optimal clustering iff for every $k$ clustering of the dataset $X$ whose cost is at most $(1+\gamma) \cdot \OPT$, the partition disagrees with the optimal clustering on at most $\delta$ fraction of the points.
\end{definition}

The following result from ~\cite{abs10} shows that if the given dataset $X$ is $(1+\gamma, \delta)$ approximation stable and all clusters in the optimal clustering have at least $\delta \cdot |X|$ points, then the instance $(X, k)$ also satisfies $(1+\gamma)$-weak deletion stability.

\begin{lemma}[Claim 3.4 in \cite{abs10}]
Any $k$-means instance $(X, k)$ that satisfies $(1+\gamma, \delta)$-approximation stability with respect to an optimal clustering and that has all optimal clusters of size at least $\delta |X|$, also satisfies the $(1+\gamma)$-weak deletion stability.
\end{lemma}

Finally, we present the weakest separation condition that has been formulated with respect to cost-based notions. 
This separation condition is known as $\beta$-distributed and was defined by Awasthi \etal~\cite{abs10}.
This definition is given below.

\begin{definition}[$\beta$-distributed]
A $k$-means instance $(X, k)$ is called $\beta$-distributed iff the following holds for any optimal clustering $\{X^{\star}_1, ..., X^{\star}_k\}$:
\[
\forall i, \forall x \notin X^{\star}_i, ||x - \mu(X^{\star}_i)||^2 \geq \beta \cdot \frac{\OPT}{|X_i^{\star}|}.
\]
\end{definition}

Awasthi \etal~\cite{abs10} showed that the above property is weaker than the weak-deletion stability property. 

\begin{lemma}[Theorem 3.5 in \cite{abs10}]
Any $(1 + \gamma)$-weakly deletion stable $k$-means instance is also $(\gamma/4)$-distributed.
\end{lemma}

The relationship between the separation notions is shown in Figure~\ref{fig:1}. 
Note that $\beta$-distributed is the weakest cost-based separation notion.
This basically means that a positive result for $\beta$-distributed instances also implies positive results with respect to other separation notions.
So any strong result with respect to the notion of $\beta$-distribution will supersede known results with respect to other notions. 
We also show that the implication from $(1+\gamma)$-weak deletion stability to 
$O(\gamma)$-distributed is strict, i.e., there exist instances which are $\gamma$-distributed, but not $(1+\Omega(\gamma))$-weak deletion stable. 
This is formally stated as the theorem below the proof of which can be found in Appendix~\ref{app:stab}.

\begin{figure}
\centering
\includegraphics[scale=0.25]{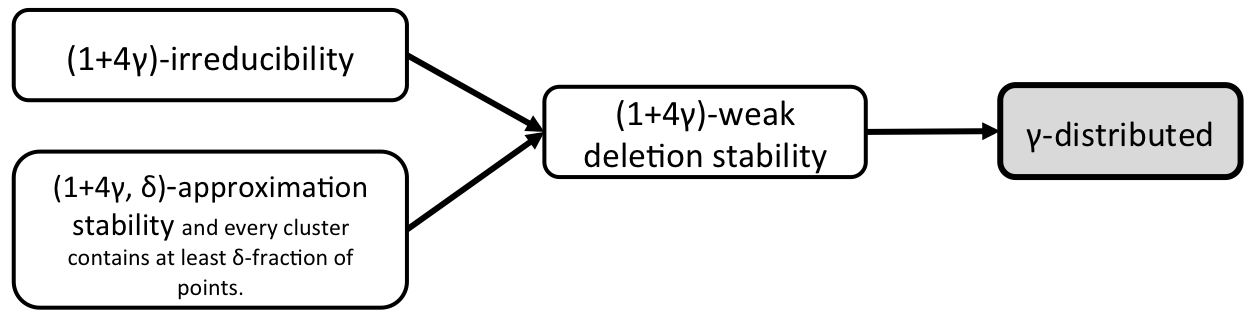}
\caption{The figure shows the relationship between various cost-based separation notions. The arrows from $A$ to $B$ denotes that if the dataset satisfies $A$, then it also satisfies $B$.}
\label{fig:1}
\end{figure}

\begin{theorem}
\label{thm:weakdel}
There are instances $(X,k)$ of the $k$-means problem which are 
$\gamma$-distributed, but are not $(1+\Omega(\gamma))$-weak deletion stable for some parameter $\gamma > 0$. 
\end{theorem}

\subsection{Algorithm and analysis}

In the previous subsection we saw that of all the center-based separation notions, $\beta$-distribution stability condition is the weakest.
Hence any result for datasets satisfying the $\beta$-distribution stability condition will have consequences for datasets satisfying stronger conditions.
So, the question is {\it are there good algorithms for datasets under this condition?} 
Awasthi \etal~\cite{abs10} gave a {PTAS} for the $k$-means/median problems on datasets that satisfy the $\beta$-distributed assumption. 
The running time has polynomial dependence on the input parameters $n, k, d$ (where $n$ is the number of data points, $d$ is the dimension of the dataset, and $k$ is the number of clusters) and exponential dependence on $\frac{1}{\beta}$ and $\frac{1}{\veps}$ ($\veps$ is the accuracy parameter). 
Even though they showed that the super-polynomial dependence on $\frac{1}{\beta}$ and $\frac{1}{\veps}$ cannot be improved,  improving the dependence on other input parameters was left as an open problem. 
In this work, we address this open problem by giving a faster {PTAS} for the $k$-means problem under the $\beta$-distributed notion. 
The running time of the algorithm for the $k$-means problem by Awasthi \etal~\cite{abs10} is $O \left(dn^3 (k \log{n})^{\poly(\frac{1}{\beta}, \frac{1}{\veps})} \right)$.
We improve the running time to $O \left(dn^3  \left( \frac{k}{\veps} \right)^{O(\frac{1}{\beta \veps^2})} \right)$.
Note that due to our improvement in running time, our algorithm is also a Fixed Parameter Tractable Approximation Scheme (FPT-AS) for the problem with parameters are $k$ and $\beta$. 
Moreover, the running time does not have an exponential dependence on $k$ that is typically the case for such FPT approximation schemes for general datasets.
We formally state our result as the following theorem which is a restatement of Theorem~\ref{thm:fpt-as-main}.

\begin{theorem}
Let $\veps, \beta > 0$, $k$ be a positive integer, and let $X \subset \R^d$ be a $\beta$-distributed dataset. 
There is an algorithm that takes as input $(X, k, \veps, \beta)$ and outputs a $k$-center set $C$ such that $\Phi(C, X) \leq (1 + \veps) \cdot \OPT$ and the algorithm runs in time $O \left(dn^3  \left( \frac{k}{\veps} \right)^{O(\frac{1}{\beta \veps^2})} \right)$.
\end{theorem}

In the remainder of this subsection, we will discuss this algorithm and its running time.
We will use the definitions related to the $k$-means problem as given in the preliminaries.
We will assume for this discussion that the value $\OPT$ is known. 
We will use that standard doubling argument to support this assumption.
We obtain our PTAS building on the ideas of Awasthi \etal~\cite{abs10}. 
One of the main definitions from \cite{abs10} that will be useful is that of {\em expensive} and {\em cheap} clusters.

\begin{definition}[Expensive and cheap clusters]
A cluster $\Xst_i$ is said to be cheap if $$\Delta(\Xst_i) \leq \frac{(\beta \veps) \cdot OPT}{4^6}$$ and is said to be expensive otherwise.
\end{definition}

\noindent
From a simple averaging argument, we get that the number of expensive clusters cannot exceed $\frac{4^6}{\beta \veps}$.
One of the main ideas in \cite{abs10} is to consider expensive and cheap clusters separately. 
For the expensive clusters the good centers are chosen in a more brute-force manner. 
Cheap clusters need to be handled more carefully. 
Our improvement in the running time mainly comes from improvement in the algorithm for the expensive clusters.
Note that in Awasthi \etal~\cite{abs10} the discovery of good centers of expensive and cheap clusters have to be interleaved because of a technicality where the expensive clusters may not get exposed unless some cheap clusters have been removed.
In our algorithm, we find good centers for the expensive clusters {\em before} finding good centers for the cheap ones. 
Since our algorithm for the expensive clusters may not deterministically give good centers for the expensive clusters, we may have to try various set of centers.
However, we will show that the number of candidates that we will have to try is not very large. 
Hence, the running time does not become too large either.
We give the algorithm and analysis next. 
Since this is built upon the previous work of Awasthi \etal (we will use their algorithm as a subroutine and use their analysis), a complete understanding of the algorithm and analysis may require an understanding of their work.

Let $t$ denote the number of expensive clusters and WLOG assume that the first $t$ clusters, that is $\Xst_1, ..., \Xst_t$, are the expensive clusters.
We know that $t \leq \frac{4^6}{\beta \veps}$.
The goal is to find a set $C_{exp}$ of good centers for the expensive clusters. 
In Section~\ref{sec:good-centers} we saw how to obtain a center set $C_{exp}$ such that the following holds:
\begin{equation}\label{eqn:fpt-as-2-1}
\Phi(C_{exp}, \{\Xst_{1}, ..., \Xst_{t}\}) \leq \left(1+\frac{\veps}{2}\right) \cdot \sum_{j=1}^{t} \Delta(\Xst_j) + \frac{\veps}{2} \cdot \OPT.
\end{equation}
To see the above, all we need to do is to set $\Xst_j = X_j$ for all $j$ and use the fact that $\Phi(C_{exp}, \{\Xst_{1}, ..., \Xst_{t}\}) \leq \psi(C_{exp}, \{\Xst_{1}, ..., \Xst_{t}\})$.
Note that the algorithm {\tt $t$-GoodCenters} returns a list $\mathcal{L}$ of candidate $t$-center sets for the expensive clusters such that at least one of the elements on this list is a set of good $t$-center set for the expensive clusters.
For each of the candidate $t$-center sets $C$ in the list, we execute the algorithm of Awasthi \etal~\cite{abs10} for the cheap clusters.
Since one of the $t$-center sets is a good set of $t$ centers for the expensive cluster (w.h.p.), we can expect that this set of $t$ centers combined with the set of centers returned by the algorithm for the cheap clusters will be good for the entire dataset. 
We need to be careful though. 
There is a non-trivial dependency on the fact that good centers for the expensive clusters have been discovered. 
That is, the analysis of the algorithm for the cheap clusters as stated in \cite{abs10} assumes that for every expensive cluster, a center close to the actual centroid has been discovered.
Such a guarantee is not given by eqn.~(\ref{eqn:fpt-as-2-1}) which only gives a guarantee in terms of the overall cost of the expensive clusters.
This means that we have to revisit the argument for the cheap clusters and make sure that the it holds even under the weaker guarantee that eqn.~(\ref{eqn:fpt-as-2-1}) gives.

We will run the algorithm for the cheap clusters (see Figure 2 in \cite{abs10}) $|\L|$ times with $\mathcal{Q}_{init}$ set to a different element of $\L$ each time and then pick the best overall solution.
This algorithm is given below.

\begin{framed}
{\tt FasterPTAS}($X, k, \veps$, $\beta$)\\
\hspace*{0.6in} {\bf Inputs}: Dataset $X$, number of clusters $k$, accuracy $\veps$, and stability parameter $\beta$\\
\hspace*{0.6in} {\bf Output}: A set $C$ of $k$ centers\\
\hspace*{0.6in} {\bf Constants}: $t = \lceil \frac{4^6}{\beta \veps} \rceil$\\
\hspace*{0.2in} (1) \ \ \  Call \GC($X, k, \veps, t$) to obtain a list $\L$ of $t$-center sets\\
\hspace*{0.2in} (2) \ \ \ For each $C_{exp} \in \L$:\\
\hspace*{0.2in} (3) \hspace*{0.2in}\ \ \ Execute the algorithm of Awasthi \etal (Figure~2 in ~\cite{abs10}) \\
\hspace*{0.9in} with $\mathcal{Q}_{init}$ set as $C_{exp}$ and get back the result $C$\\
\hspace*{0.2in} (4) \hspace*{0.2in}\ \ \ If $C$ is not $\bot$, then return($C$) else continue\\
\hspace*{0.8in} //{\it Note that the algorithm of Awasthi \etal either returns $C$ such that} \\
\hspace*{0.8in}  //{\it $\Phi(C, X) \leq (1 + \veps) \OPT$ or returns $\bot$ indicating failure to find such a center set}
\end{framed}

\paragraph{Running time}: Since the list size for parameters $k, t, \veps$ is $(\frac{k}{\veps})^{O(\frac{t}{\veps})}$ and $t = O(\frac{1}{\beta \veps})$ the running time will be $(\frac{k}{\veps})^{O(\frac{1}{\beta \veps^2})}$ times the time for execution of the algorithm for Awasthi \etal (figure 2 in \cite{abs10}). The execution of their algorithm takes time $O(d n^3 k^{O(\frac{1}{\beta})})$. So, the running time of the above algorithm is $O \left(dn^3  \left( \frac{k}{\veps} \right)^{O(\frac{1}{\beta \veps^2})} \right)$.

\vspace{0.2in}

We will show that with high probability there exists a $t$-center set $C_{exp} \in \L$ such that eqn.~(\ref{eqn:fpt-as-2-1}) holds.
However, before we prove this, let us try to see why the existence of such a $C_{exp}$ in $\L$ is sufficient to obtain a $k$-center set $C$ (w.h.p.) such that $\Phi(C, X) \leq (1 + \veps) \cdot \OPT$.
This requires going back to the analysis of Awasthi \etal~\cite{abs10}.

\paragraph{Analysis of Awasthi \etal} 
We essentially run the algorithm of Awasthi \etal for various choices of $\mathcal{Q}_{init}$. 
The algorithm in \cite{abs10} is guaranteed (w.h.p.) to output a good center set for the entire dataset under the condition that $\mathcal{Q}_{init} = \{c_1, ..., c_t\}$ such that $\Phi(c_j, \Xst_j) \approx \Phi(\cst_j, \Xst_j)$ for $j = 1, ..., t$. 
In fact, examining more closely, the following two properties are needed in the analysis:
\begin{enumerate}
\item $\Phi(\mathcal{Q}_{init}, \cup_{j=1}^{t} \Xst_j) \leq (1+ \veps) \cdot \sum_{j=1}^{t} \Delta(\Xst_j)$, and
\item $\Phi(\mathcal{Q}_{init} \cup \{\cst_{t+1}, ..., \cst_k\}, X) \leq (constant) \cdot \OPT$.
\end{enumerate}
Readers familiar with the analysis in \cite{abs10} will realise that the second property is required to bound the number of ``{\em bad components}".
Given that the second property is satisfied by $\mathcal{Q}_{init}$, the algorithm finds centers $c_{t+1}, ..., c_{k}$ such that, $\Phi(c_j, \Xst_j) \leq (1 + \veps) \cdot \Delta(\Xst_j)$ for all $j = t+1, ..., k$.
Note that when combined with property (1), we get the desired result that $\Phi(\mathcal{Q}_{init} \cup \{c_{t+1}, ..., c_{k}\}, X) \leq (1+\veps) \cdot \OPT$.

Now let us consider the guarantee that is provided by the  \GC algorithm and see whether it aligns with the analysis in Awasthi \etal.
The guarantee that is provided is that (w.h.p.) $\L$ has at least one center set $C_{exp}$ such that 
\[
\Phi(C_{exp}, \cup_{j=1}^{t} \Xst_j) \leq \left(1+ \frac{\veps}{2} \right) \cdot \sum_{j=1}^{t} \Delta(\Xst_j) + \frac{\veps}{2} \cdot \OPT.
\]
The above inequality gives:
\begin{eqnarray*}
\Phi(C_{exp} \cup \{\cst_{t+1}, ..., \cst_k\}, X) &\leq& \Phi(C_{exp}, \cup_{j=1}^{t} \Xst_j) + \sum_{j=t+1}^{k} \Phi(\cst_j, \Xst_j)\\
&\leq& \left(1+\frac{\veps}{2} \right) \left(\sum_{j=1}^{t} \Delta(\Xst_j) \right) + \frac{\veps}{2} \cdot \OPT + \sum_{j=t+1}^{k} \Delta(\Xst_j)\\
&\leq& (1+\veps) \cdot \OPT
\end{eqnarray*}
This implies that property (2) is satisfied for at least one of the element in the list $\L$. 
This further means that the algorithm of Awasthi \etal, when initiated with this element and $(\frac{\veps}{2})$ as accuracy parameter, finds centers $\{c_{t+1}, ..., c_{k}\}$ such that $\Phi(c_j, \Xst_j) \leq (1+ \frac{\veps}{2}) \cdot \Delta(\Xst_j)$ for $j = t+1, ..., k$. This gives:
\[
\Phi(C_{exp} \cup \{c_{t+1}, ..., c_k\}, X) \leq \left(1+ \frac{\veps}{2} \right) \cdot \sum_{j=1}^{t} \Delta(\Xst_j)  + \frac{\veps}{2} \cdot \OPT + \left( 1+ \frac{\veps}{2}\right) \cdot \sum_{j=t+1}^{k} \Delta(\Xst_j) \leq (1+\veps)\cdot \OPT.
\]

This completes the analysis of the algorithm and the proof of Theorem~\ref{thm:fpt-as-main}.

\section{Parallel {PTAS} for $k$-means}\label{sec:parallel}
In this section, we give a massively parallel {PTAS} for the classical and constrained $k$-means problems. 
We will first discuss the classical $k$-means problem and generalise for the constrained $k$-means.
This actually just comes from a close inspection of the algorithm \GC (note that for a PTAS we will use $t=k$). 
One quickly realises that most of the steps in the algorithm can be performed independently and hence the algorithm can easily be converted to a massively parallel {PTAS} for the $k$-means problem. 
The main results of this section is given in the theorem below which is a restatement of Theorem~\ref{thm:parallel-main}.

\begin{theorem}
Let $\veps > 0$, $(X, k, d)$ be a $k$-means instance, and let $C$ denote a constant $\alpha$-approximate solution of the $k$-means instance. 
Then there is a parallel algorithm in the shared memory CREW model that takes as input the $k$-means instance, $C$, and $\veps$ and outputs a $(1+\veps)$-approximate solution in parallel time $O \left( \left\lceil \frac{nd 2^{\tilde{O}(k/\veps)}}{N}\right\rceil + \frac{k}{\veps} \log{\frac{k}{\veps}}+\log{(nkd)} \right)$ with $N$ processors.
There is similar parallel algorithm in the CRCW model with running time $O \left( \left\lceil \frac{nd 2^{\tilde{O}(k/\veps)}}{N}\right\rceil + \frac{1}{\veps} +\log{(nkd)} \right)$.
For any constrained version of the $k$-means problem with partition algorithm $\mathcal{P}$, there is a parallel algorithm with running time $O \left( \left\lceil \frac{nd 2^{\tilde{O}(k/\veps)}}{N}\right\rceil \cdot t(n,k,d)+\frac{k}{\veps} \log{\frac{k}{\veps}} + \log{(nkd)} \right)$ in the CREW model with $N$ processors. Here, $t(.)$ denotes the running time of the partition algorithm $\mathcal{P}$.
\end{theorem}

Let us look at the algorithm \GC closely to see the possibility of performing each of the steps in parallel. 
We discuss with respect to the shared memory model.

\begin{framed}
\GC($X, C, \veps, t$)\\
\hspace*{0.6in} {\bf Inputs}: Dataset $X$, $(\alpha, \beta)$-approximate $C$, accuracy $\veps$, and number of centers $t$\\
\hspace*{0.6in} {\bf Output}: A list $\L$ of $t$ elements, each element being a set of $t$ centers\\
\hspace*{0.6in} {\bf Constants}: $\eta = \frac{2^{16} \alpha t}{\veps^4}; \tau = \frac{128}{\veps}$\\
\hspace*{0.2in} (1) \ \ \ $\L \leftarrow \emptyset$\\
\hspace*{0.2in} (2) \ \ \ Repeat $2^t$ times:\\
\hspace*{0.2in} (3)\hspace*{0.3in}  \ \ \ Sample a multi-set $M$ of $\eta t$ points from $X$ using $D^2$-sampling w.r.t. center set $C$\\
\hspace*{0.2in} (4)\hspace*{0.3in}  \ \ \ $M \leftarrow M \cup$ \{$\frac{128 t}{\veps}$ copies of each element in $C$\}\\
\hspace*{0.2in} (5)\hspace*{0.3in} \ \ \ For all disjoint subsets $S_1, ..., S_t$ of $M$ such that $\forall i, |S_i| = \tau$:\\
\hspace*{0.2in} (6)\hspace*{0.9in} $\L \leftarrow \L \cup (\mu(S_1), ..., \mu(S_t))$\\
\hspace*{0.2in} (7) \ \ \ return($\L$)
\end{framed}

Let us discuss each of the steps in the above algorithm:

\begin{itemize}
\item {\it Step (1)}: This is a trivial step.

\item {\it Step (2)}: This is an iteration of size $2^t$ which is for probability amplification. The rounds are independent and can be performed in parallel. 

\item {\it Step (3)}: Since the points are $D^2$-sampled independently, it should be possible to execute this step in parallel. For this we need to compute the distribution for $D^2$-sampling w.r.t. $C$ in parallel. What we can do is first compute $\min_{c \in C} ||x - c||^2$ for every point $x \in X$ in parallel and then aggregate the smallest distances to compute the distribution. 
There are logarithmic aggregation costs involved.
The cost of aggregating across $d$ dimensions for calculating Euclidean distance is $\log{d}$, that of finding the distance of a point to the closest center in $C$ is $\log{k}$, and aggregating these distance has cost $\log{n}$.
Once the distribution has been computed the sampling can be done in parallel.
In summary, if there are $N$ processors, then the parallel running time for this step in the CREW model will be $O \left(\left\lceil \frac{ndk\eta t}{N} \right\rceil + \log{(nkd)} \right)$.

\item {\it Step (4)}: This is a simple step and can be performed in parallel. 

\item {\it Step (5-6)}: The disjoint subsets can be considered in parallel. 
Taking into consideration the aggregation costs, the parallel running time using $N$ processors is $O\left( \left\lceil \frac{t\tau d \L}{N} \right\rceil + \log{(t \tau d)} \right)$. 

\item {\it Step (7)}: This is trivial.
\end{itemize}

Let us now discuss about the parallel PTAS.
Firstly, note that for a PTAS we will use $t=k$ in the \GC algorithm.
Secondly, note that the above algorithm returns a list $\L$. 
A PTAS should return a single $k$-center set.
So, we have to address the issue of finding the $k$-center set in the list $\L$ with the least $k$-means cost in parallel time. 
This is an aggregation step and in the CREW model will have a $\log{\L}$ cost.
The parallel running time for finding the best solution from a list of size $\L$ using $N$ processors in the CREW model will be $O \left( \left\lceil \frac{nkd \L}{N}\right\rceil + \log{(nkd \L)}\right)$. 
So the overall running time of the parallel algorithm is $O \left( \left\lceil \frac{nd 2^{\tilde{O}(k/\veps)}}{N}\right\rceil +  \frac{k}{\veps} \log{\frac{k}{\veps}}+\log{(nkd)}\right) $.
Note that the parallel running time cannot be better than $O(\frac{k}{\veps})$ since $\L  = 2^{\tilde{O}(\frac{k}{\veps})}$.
So, even though we have removed the $k$-sized iteration compared to the parallel algorithm of Jaiswal \etal (Algorithm 2 in \cite{jks}), we do not obtain a parallel running time improvement over \cite{jks} that has a similar running time.
The main bottleneck is the aggregation over the list $\L$ in the CREW model.
Interestingly, in the CRCW model we can do something better.

Note that we do not really care about the $k$-center set in $\L$ with the least cost. 
What we want is a $k$-center set $\mathcal{C}$ such that $\Phi(\mathcal{C}, X) \leq (1 + \veps) \cdot \OPT$.
Next, we describe how this can be obtained in the CRCW model.
Let $\Lambda = \Phi(C, X)$. 
That is $\Lambda$ is the cost of the $\alpha$-approximate solution that is used in the \GC algorithm. Given this we know that for the constant $\alpha$, $\Lambda \leq \alpha \cdot \OPT$.
Consider the following ranges of $k$-means cost: 
\begin{equation*}
S_0 = ((1-\veps) \Lambda, \Lambda], S_1 = ((1-\veps)^2 \Lambda, (1-\veps) \Lambda], ..., S_i = ((1-\veps)^{i+1} \Lambda, (1-\veps)^i \Lambda], ...
\end{equation*}
Since $\OPT$ is the smallest cost of a solution, the number of ranges that we need to consider is $O(\frac{\log{1/\alpha}}{\log{1 - \veps}}) = O(\frac{\log{\alpha}}{\veps})$.
Suppose there are $\L$ processors, then each processor can calculate the cost of one of the $k$-center sets in the list and then using the CRCW shared memory indicate which of the ranges the cost belongs to. 
Eventually the least range that is populated will give the $(1+\veps)$-approximate solution.
So, the parallel running time with $N$ processors in the CRCW model is $O \left( \left\lceil \frac{nd 2^{\tilde{O}(k/\veps)}}{N}\right\rceil +  \frac{1}{\veps} +\log{(nkd)}\right)$ which is much better than that in the CREW model.

\paragraph{Parallel algorithms for constrained $k$-means} The parallel algorithm in the CREW model is the same for the constrained variations of the $k$-means problem. The main difference for various constrained variations is in the last step where an appropriate $k$-center set should be chosen from the list $\L$. The {\em partition algorithm} corresponding to the particular constrained variation is used to pick a $k$-center set. 
Suppose the running time of the partition algorithm is $t(n, k, d)$. 
Then the parallel running time in the CREW model is $O \left( \left\lceil \frac{nd 2^{\tilde{O}(k/\veps)}}{N}\right\rceil \cdot  t(n,k,d) +\log{(nkd \L)}\right)$.
The trick of the CRCW model does not extend to the constrained setting since we only have a constant factor approximate solution for the $k$-means problem which is not necessarily a constant factor approximate solution for the constrained $k$-means problem.

\paragraph{Parallel constant-approximation for $k$-means} It is important to note that the parallel algorithms discussed in this section only convert a constant factor approximate solution $C$ for the $k$-means problem to a $(1+\veps)$-approximate solution. 
This is because the \GC algorithm needs a center set $C$ such that $\Phi(C, X) \leq \alpha \cdot \OPT$.
For standalone parallel PTAS, we also need to design a parallel algorithm for finding such a constant factor approximate solution. 
We design such a parallel algorithm using the ideas of Guha \etal~\cite{guha03}. 
The main idea is captured in the following lemma from ~\cite{guha03}.

\begin{lemma}
Suppose for a given dataset $X$, there is a point set $X'$ such that $\Phi(X', X) \leq \beta \cdot \OPT$ and let $X''$ denotes the weighted set of points that are the same as $X'$ and weighted as per the Voronoi partitioning of $X$ with respect to $X'$. If $\Phi(C, X'') \leq \beta' \cdot \OPT(X'')$ for some point set $C$, then $\Phi(C, X) \leq (2\beta + 4\beta'(\beta + 1)) \cdot \OPT$.
\end{lemma}

The above lemma can be used to design the following simple parallel algorithm in the CREW model with $\sqrt{n/k}$ processors. We do this simple case since it is easy to describe and later generalise further. 
Let $A$ be a constant $\beta$-factor approximation algorithm for the $k$-means problem. 
Such constant factor approximation algorithms that run in polynomial time are known~\cite{ahmadian17,arya01}. 
Consider a partition of the dataset $X$ into $\sqrt{n/k}$ partitions $S_1, ..., S_{\sqrt{n/k}}$ each containing $\sqrt{nk}$ points.
In the parallel algorithm, processor $i$ uses algorithm $A$ on the partition $S_i$ to find $k$ centers $C_i$ for its partition. 
Processor $i$ also does the Voronoi partitioning of $S_i$ with respect to $C_i$ and computes the weighted sets $C_1', ..., C_{\sqrt{n/k}}'$.
After this parallel step the algorithm $A$ is used on the collected set $\cup_{i=1}^{\sqrt{n/k}} C_i'$ of $k \cdot \sqrt{n/k} = \sqrt{nk}$ points to obtain a center set $\mathcal{C}$. This is produced as the solution. 
Note that this is an $O(1)$ approximation algorithm from the above lemma of \cite{guha03}.
Suppose the running time of algorithm $A$ for input parameters $n, k, d$ is $t(n, k, d)$. Then the running time of the parallel algorithm is $t(\sqrt{nk}, k, d)$.
The above algorithm is a two-level algorithm. We can extend this idea to design a multi-level algorithm where each processor works on small subset of weighted points. 
This idea has been used in the past for designing small-space streaming algorithms~\cite{guha03,ajm09}. We give the final result below. The detailed description of the construction may be found in \cite{guha03,ajm09}.

\begin{theorem}
Let $0 < \veps < 1$. There is a parallel algorithm in the CREW model with $N$ processors that runs in time $poly(n^{\veps}, k, d, \frac{1}{\veps}) \cdot \left\lceil \frac{n^{1-\veps}}{N}\right\rceil$ and outputs a constant $c$-factor solution where $c = \tilde{c}^{1/\veps}$ for some global constant $\tilde{c}$.
\end{theorem}

\section{Conclusion and open problems}
Our results demonstrate the versatility of the sampling based approach in the context of the $k$-means problem. 
This has also been demonstrated in some of the past works.
The effectiveness of $k$-means++ (which is basically $D^2$-sampling in $k$ rounds) is well known~\cite{av07}.
The $D^2$-sampling technique has been used to give simple {PTAS} for versions of the $k$-means/median problems with various metric-like distance measures~\cite{jks} and also various constrained variations of $k$-means~\cite{bjk}.
It has also been used to give efficient algorithms in the semi-supervised setting~\cite{abjk,ghs18} and coreset construction~\cite{lfkf17}.
In this work, we see its use in the streaming, parallel, and clustering-under-stability settings.
The nice property of the sampling based approach is that we have a uniform template of the algorithm that is extremely simple and that works in various different settings. 
This essentially means that the algorithm remains the same while the analysis changes. 

There are multiple problems that are left open in this work. 
In the streaming setting, we give a generic algorithm within the unified framework of Ding and Xu~\cite{dx15}. 
The advantage of working in this unified framework is that we get streaming algorithms for various constrained versions of the $k$-means problem. 
However, it may be possible to obtain better streaming algorithms for the constrained problems when considered separately. 
For instance, our streaming algorithm works in $4$ passes. 
So, one important question is whether it is possible to design a single-pass algorithm. 
The running time of our streaming algorithm has an exponential dependence on $k$ which is not a problem as long as $k$ is a fixed constant and not part of the input. 
If that is not the case, then our algorithm is not very efficient. 
Moreover, there is an efficient constant-approximation factor streaming algorithm for the classical $k$-means problem~\cite{brav11}. 
So, the relevant question is whether such efficient algorithms can also be designed for the constrained versions of the $k$-means problem.

For the FPT approximation scheme, we obtained an algorithm with running time $O(dn^3 \cdot f(params))$. 
The relevant open problem in this case is whether a linear time (i.e., $O(nd \cdot f(params))$) FPT approximation scheme can be achieved.
An important observation related to our FPT approximation scheme is that our algorithm simply follows the generalisation of the $D^2$-sampling based algorithm that one can find a list of good $t$-center sets for any fixed set of $t$ clusters. 
This allowed us to deal with a few clusters  which were called expensive clusters in the setting considered in this work). 
There may be other settings where apart from a few ``bad" clusters, it is easy to find good centers for the rest of the clusters.
So, an interesting question is whether there are other problem instances where this idea can be exploited.

\paragraph{\bf Acknowledgements.} The authors thank Anup Bhattacharya for helpful discussions. The authors would also like to thank Sanjeev Khanna and Sepehr Assadi for allowing us to use their impossibility argument for the chromatic $k$-means problem.

\addcontentsline{toc}{section}{References}
\bibliographystyle{alpha}
\bibliography{paper}

\appendix
\section{Proof of Lemma~\ref{lem:list-mainlem} (continued)}\label{app:list-mainlem}

We continue with the proof of Lemma~\ref{lem:list-mainlem}.

\paragraph{\bf Case-I: $\left(\Phi(C, X_j) \leq \frac{\veps}{6 \alpha t} \cdot \Phi(C, X) \right)$}

First, note that the following follows from the the fact that $\Phi(C, X) \leq \alpha \cdot \OPT$ (eqn. (\ref{eqn:cost})) and $\OPT \leq OPT$:
\begin{equation}\label{eqn:case1-1}
\Phi(C, X_j) \leq \frac{\veps}{6t} \cdot OPT
\end{equation}
For any point $x \in X$, let $c(x)$ denote the center in the set $C$ that is closest to $x$. That is, $c(x) = \arg\min_{c \in C}{||c - x||}$.
Given this definition, note that:
\begin{equation}
\sum_{x \in X_j} ||x - c(x)||^2 = \Phi(C, X_j)
\end{equation}
We define the multi-set $X_j' = \{c(x) : x \in X_j\}$.
Let $m$ and $m'$ denote the means of the point sets $X_j$ and $X_j'$ respectively. 
So, we have $\Delta(X_j) = \Phi(m, X_j)$ and $\Delta(X_j') = \Phi(m', X_j')$.
We will show that $\Delta(X_j) \approx \Delta(X_j')$. 
First, we bound the distance between $m$ and $m'$.

\begin{lemma}\label{lemma:case1-1}
$||m - m'||^2 \leq \frac{\Phi(C, X_j)}{|X_j|}$.
\end{lemma}
\begin{proof}
We have:
\begin{eqnarray*}
||m - m'||^2 = \frac{\lvert\lvert \sum_{x \in X_j} (x - c(x))\rvert\rvert^2}{|X_j|^2} \leq  \frac{\sum_{x \in X_j} \lvert\lvert (x - c(x))\rvert\rvert^2}{|X_j|} = \frac{\Phi(C, X_j)}{|X_j|}.
\end{eqnarray*}
where the second last inequality follows from Cauchy-Schwartz.\qed
\end{proof}

\begin{lemma}\label{lemma:case1-2}
$\Delta(X_j') \leq 2 \cdot \Phi(C, X_j) + 2 \cdot \Delta(X_j)$.
\end{lemma}
\begin{proof}
We have:
\begin{eqnarray*}
\Delta(X_j') &=& \sum_{x \in X_j} ||c(x) - m'||^2 \leq \sum_{x \in X_j} ||c(x) - m||^2 \\
&\stackrel{\tiny{(Fact\ \ref{fact:2})}}{\leq}& 2 \cdot \sum_{x \in X_{j}} (||c(x) - x||^2 + ||x - m||^2) = 2 \cdot \Phi(C, X_j) + 2 \cdot \Delta(X_j)
\end{eqnarray*}
This completes the proof of the lemma.\qed
\end{proof}

We now show that a good center for $X_j'$ will also be a good center for $X_j$.

\begin{lemma}
Let $m''$ be a point such that $\Phi(m'', X_j') \leq (1 +\frac{\veps}{8}) \cdot \Delta(X_j')$.
Then $\Phi(m'', X_j) \leq (1 + \frac{\veps}{2}) \cdot \Delta(X_j) + \frac{\veps}{2t} \cdot OPT$.
\end{lemma}
\begin{proof}
We have:
\begin{eqnarray*}
\Phi(m'', X_j) &\stackrel{\tiny{(Fact~\ref{fact:1})}}{=}& \sum_{x \in X_j} ||x - m||^2 + |X_j| \cdot ||m - m''||^2\\
&\stackrel{\tiny{(Fact~\ref{fact:2})}}{\leq}& \Delta(X_j) + 2|X_j| \cdot (||m - m'||^2 + ||m' - m''||^2) \\
&\stackrel{\tiny{(Lemma~\ref{lemma:case1-1})}}{\leq}& \Delta(X_j) + 2 \cdot \Phi(C, X_j) + 2 |X_j| \cdot ||m' - m''||^2\\
&\stackrel{\tiny{(Fact~\ref{fact:1})}}{\leq}& \Delta(X_j) + 2 \cdot \Phi(C, X_j) + 2 (\Phi(m'', X_j') - \Delta(X_j'))\\
&\leq&  \Delta(X_j) + 2 \cdot \Phi(C, X_j) + \frac{\veps}{4} \cdot \Delta(X_j')\\
&\stackrel{\tiny{(Lemma~\ref{lemma:case1-2})}}{\leq}&  \Delta(X_j) + 2 \cdot \Phi(C, X_j) + \frac{\veps}{2} \cdot (\Phi(C, X_j) + \Delta(X_j)) \\
&\stackrel{\tiny{(Eqn.~\ref{eqn:case1-1})}}{\leq}& \left(1+\frac{\veps}{2} \right) \cdot \Delta(X_j) + \frac{\veps}{2t} \cdot OPT.
\end{eqnarray*}
This completes the proof of the lemma.\qed
\end{proof}

We know from Lemma~\ref{lemma:inaba} that there exists a (multi) subset of $X_j'$ of size $\frac{16}{\veps}$ such that the mean of these points satisfies the condition of the lemma above.
Since $C_j'$ contains at least $\frac{16}{\veps}$ copies of every element of $C$, there is guaranteed to be a subset $T_j \subseteq C_j'$ that satisfies eqn. (\ref{eqn:reqd}).
So, for any index $j \in \{1, ..., t\}$ such that $\frac{\Phi(C, X_j)}{\Phi(C, X)} \leq \frac{\veps}{6 \alpha t}$, $M_j$ has a good subset $T_j$ with probability $1$.

\paragraph{\bf Case-II: $\left(\Phi(C, X_j) > \frac{\veps}{6 \alpha t} \cdot \Phi(C, X) \right)$}

If we can show that a $D^2$-sampled set with respect to center set $C$ has a subset $S$ that may be considered uniform sample from $X_j$, then we can use Lemma~\ref{lemma:inaba} to argue that $M_j$ has a subset $T_j$ such that $\mu(T_j)$ is a good center for $X_j$. 
Note that since $\frac{\Phi(C, X_j)}{\Phi(C, X)} > \frac{\veps}{6 \alpha t}$, we can  argue that if we $D^2$-sample $poly(\frac{t}{\veps})$ elements, then we will get a good representation from $X_j$.
However, some of the points from $X_j$ may be very close to one of the centers in $C$ and hence will have a very small chance of being $D^2$-sampled. 
In such a case, no subset $S$ of a $D^2$-sampled set will behave like a uniform sample from $X_j$. 
So, we need to argue more carefully taking into consideration the fact that there may be points in $X_j$ for which the chance of being $D^2$-sampled may be very small. 
Here is the high-level argument that we will build:
\begin{itemize}
\item Consider the set $X_j'$ which is same as $X_j$ except that points in $X_j$ that are very close to $C$ have been ``collapsed" to their closest center in $C$.
\item Argue that a good center for the set $X_j'$ is a good center for $X_j$.
\item Show that a convex combination of copies of centers in $C$ (i.e., $C_j'$) and $D^2$-sampled points from $X_j$ gives a good center for the set $X_j'$.
\end{itemize}
The closeness of point in $X_j$ to points in $C$ is quantified using radius $R$ that is defined by the equation: 
\begin{equation}\label{eqn:defineR}
R^2 \stackrel{defn.}{=} \frac{\veps^2}{41} \cdot \frac{\Phi(C, X_j)}{|X_j|}.
\end{equation}
Let $X_j^{near}$ be the points in $X_j$ that are within a distance of $R$ from a point in set $C$ and $X_j^{far} = X_j \setminus X_j^{near}$. That is,
$
X_j^{near} = \{x \in X_j : \min_{c \in C}{||x - c||} \leq R\}$ and $X_j^{far} = X_j \setminus X_{j}^{near}.
$
Using these we define the multi-set $X_j'$ as:
\[
X_j' = X_j^{far} \cup \{c(x) : x \in X_j^{near}\}
\]

\newcommand{\nst}{\bar{n}}

Note that $|X_j| = |X_j'|$.
Let $m = \mu(X_j)$, $m' = \mu(X_j')$.
Let $n = |X_j|$ and $\nst = |X_j^{near}|$.
We first show a lower bound on $\Delta(X_j)$ in terms of $R$.

\begin{lemma}\label{lemma:case2-1}
$\Delta(X_j) \geq \frac{16 \nst}{\veps^2} R^2$.
\end{lemma}
\begin{proof}
Let $c = \arg\min_{c' \in C}{||m - c'||}$. We do a case analysis:
\begin{enumerate}
\item \underline{Case 1}: $||m - c|| \geq \frac{5}{\veps} \cdot R$\\
Consider any point $p \in X_j^{near}$. From triangle inequality, we have:
\[
||p - m|| \geq ||c(p) - m|| - ||c(p) - p|| \geq \frac{5}{\veps} \cdot R - R \geq \frac{4}{\veps} \cdot R.
\]
This gives: $\Delta(X_j) \geq \sum_{p \in X_j^{near}} ||p - m||^2 \geq \frac{16 \nst}{\veps^2} \cdot R^2$.

\item \underline{Case 2}: $||m - c|| < \frac{5}{\veps} \cdot R$\\
In this case, we have:
\begin{eqnarray*}
\Delta(X_j) = \Phi(c, X_j) - n \cdot ||m - c||^2 \geq \Phi(C, X_j) - n \cdot ||m - c||^2 
\geq \frac{41 n}{\veps^2} \cdot R^2 - \frac{25 n}{\veps^2} \cdot R^2 \geq \frac{16 \nst}{\veps^2} \cdot R^2.
\end{eqnarray*}
\end{enumerate}
This completes the proof of the lemma.\qed
\end{proof}

We now bound the distance between $m$ and $m'$ in terms of $R$.

\begin{lemma}\label{lemma:case2-2}
$||m - m'||^2 \leq \frac{\nst}{n} \cdot R^2$.
\end{lemma}
\begin{proof}
Since $|X_j| = |X_j'|$ and the only difference between $X_j$ and $X_j'$ are the points corresponding to $X_j^{near}$, we have:
\[
||m - m'||^2 = \frac{1}{(n)^2} \left| \left| \sum_{p \in X_j^{near}} (p - c(p))\right| \right|^2 \leq \frac{\nst}{(n)^2} \sum_{p \in X_j^{near}} ||p - c(p)||^2 \leq \frac{\nst^2}{(n)^2}R^2 \leq \frac{\nst}{n} R^2.
\]
The second inequality above follows from the Cauchy-Schwarz inequality. \qed
\end{proof}

We now show that $\Delta(X_j)$ and $\Delta(X_j')$ are close.

\begin{lemma}\label{lemma:case2-3}
$\Delta(X_j') \leq 4 \nst R^2 + 2 \cdot \Delta(X_j)$.
\end{lemma}
\begin{proof}
The lemma follows from the following sequence of inequalities:
\begin{eqnarray*}
\Delta(X_j') &=& \sum_{p \in X_j^{near}} ||c(p) - m'||^2 + \sum_{p \in X_j^{far}} ||p - m'||^2\\
&\stackrel{\tiny{(Fact~\ref{fact:2})}}{\leq}& \sum_{p \in X_j^{near}} 2 \cdot \left( ||c(p) - p||^2 + ||p - m'||^2\right) + \sum_{p \in X_j^{far}} ||p - m'||^2\\
&\leq& 2 \nst R^2 + 2 \cdot \Phi(m', X_j)\\
&\stackrel{\tiny{(Fact~\ref{fact:1})}}{\leq}& 2 \nst R^2 + 2 \cdot (\Phi(m, X_j) + n \cdot ||m - m'||^2) \\
&\stackrel{\tiny{(Lemma~\ref{lemma:case2-2})}}{\leq}& 4 \nst R^2 + 2 \cdot \Delta(X_j)
\end{eqnarray*}
This completes the proof of the lemma.\qed
\end{proof}

We now argue that any center that is good for $X_j'$ is also good for $X_j$.
\begin{lemma}\label{lemma:case2-4}
Let $m''$ be  such that $\Phi(m'', X_j') \leq \left(1 + \frac{\eps}{16} \right) \cdot \Delta(X_j')$.
Then $\Phi(m'', X_j) \leq \left(1 + \frac{\eps}{2} \right) \cdot \Delta(X_j)$.
\end{lemma}

\begin{proof}
The lemma follows from the following inequalities:
\begin{eqnarray*}
\Phi(m'', X_j) &=& \sum_{p \in X_j} ||m''-p||^2 \\
&\stackrel{\tiny{(Fact~\ref{fact:1})}}{=}& \sum_{p \in X_j} ||m - p||^2 +  n \cdot ||m-m''||^2 \\
&\stackrel{\tiny{(Fact~\ref{fact:2})}}{\leq}& \Delta(X_j) + 2n \left( ||m - m'||^2 + ||m'-m''||^2 \right)  \\
&\stackrel{\tiny{(Lemma~\ref{lemma:case2-2})}}{\leq}& \Delta(X_j) + 2 \nst R^2 + 2n \cdot ||m'-m''||^2 \\
&\stackrel{\tiny{(Fact~\ref{fact:1})}}{\leq}& \Delta(X_j) + 2 \nst R^2 + 2 \cdot \left( \Phi(m'', X_j') - \Delta(X_j')\right)\\
&\stackrel{\tiny{(Lemma\ hypothesis)}}{\leq}& \Delta(X_j) + 2 \nst R^2 + \frac{\eps}{8} \cdot \Delta(X_j') \\
&\stackrel{\tiny{(Lemma~\ref{lemma:case2-3})}}{\leq}& \Delta(X_j) + 2 \nst R^2 + \frac{\eps}{2} \cdot \nst R^2 + \frac{\eps}{4} \cdot \Delta(X_j)\\
&\stackrel{\tiny{(Lemma~\ref{lemma:case2-1})}}{\leq}& \left(1 + \frac{\eps}{2} \right) \cdot \Delta(X_j).
\end{eqnarray*}
This completes the proof of the lemma.
\qed
\end{proof}

Given the above lemma, all we need to argue is that our algorithm indeed considers a center $m''$ such that $\Phi(m'', X_j') \leq (1+\veps/16) \cdot \Delta(X_j')$.
For this we would need about $\Omega(\frac{1}{\veps})$ uniform samples from $X_j'$. 
However, our algorithm can only sample using $D^2$-sampling w.r.t. $C$. 
For ease of notation, let $c(X_j^{near})$ denote the multi-set $\{c(p): p \in X_j^{near}\}$.
Recall that $X_j'$ consists of $X_j^{far}$ and $c(X_j^{near})$.
The first observation we make is that the probability of sampling an element from $X_j^{far}$ is reasonably large (proportional to $\frac{\veps}{k}$). 
Using this fact, we show how to sample from $X_j'$ (almost uniformly). 
Finally, we show how to convert this almost uniform sampling to uniform sampling (at the cost of increasing the size of sample).

\begin{lemma}
\label{lem:osample}
Let $x$ be a sample from $D^2$-sampling w.r.t. $C$.
Then, $\pr[x \in X_j^{far}] \geq \frac{\eps}{8 \alpha t}$.
Further, for any point $p \in X_j^{far}$, $\pr[x=p] \geq \frac{\gamma}{|X_j|}$, where $\gamma$ denotes $\frac{\veps^3}{246 \alpha t}$.
\end{lemma}

\begin{proof}
Note that $\sum_{p \in X_j^{near}} \pr[x=p] \leq \frac{R^2}{\Phi(C, X)} \cdot |X_j| \leq \frac{\eps^2}{41} \cdot \frac{\Phi(C, X_j)} {\Phi(C, X)}$.
Therefore, the fact that we are in case~II implies that:
$$\pr[x \in X_j^{far}] \geq \pr[x \in X_j] - \pr[x \in X_j^{near}] \geq \frac{\Phi(C, X_j)}{\Phi(C, X)} - \frac{\eps^2}{41} \cdot \frac{\Phi(C, X_j)} {\Phi(C, X)} \geq \frac{\eps}{8 \alpha t}.$$

\noindent
Also, if $x \in X_j^{far}$, then $\Phi(C, \{x\}) \geq R^2=\frac{\eps^2}{41} \cdot \frac{\Phi(C, X_j)}{|X_j|}$.
Therefore,
$$\frac{\Phi(C, \{x\})}{\Phi(C, X)}  \geq \frac{\eps}{6 \alpha t} \cdot \frac{R^2}{\Phi(C, X_j)} \geq \frac{\veps}{6 \alpha t} \cdot \frac{\veps^2}{41} \cdot \frac{1}{|X_j|} \geq \frac{\veps^3}{246 \alpha t} \cdot \frac{1}{|X_j|}.
$$
This completes the proof of the lemma.
\qed
\end{proof}

Let $O_1, \ldots O_{\eta}$ be $\eta$ points sampled independently using $D^2$-sampling w.r.t. $C$.
We construct a new set of random variables $Y_1, \ldots, Y_{\eta}$.
Each variable $Y_u$ will depend on $O_u$ only, and will take values either in $X_j'$ or will be $\nl$.
These variables are defined as follows: if $O_u \notin X_j^{far}$, we set $Y_u$ to  $\nl$. 
Otherwise, we assign $Y_u$ to one of the following random variables with equal probability:
(i) $O_u$ or (ii) a random element of the multi-set $c(X_j^{near})$.
The following observation follows from Lemma~\ref{lem:osample}.

\begin{corollary}
\label{cor:osample}
For a fixed index $u$, and an element $x \in X_j'$, $\pr[Y_u=x] \geq \frac{\gamma'}{|X_j'|},$ where $\gamma'=\gamma/2$.
\end{corollary}

\begin{proof}
If $x \in X_j^{far}$, then we know from Lemma~\ref{lem:osample} that $O_u$ is $x$ with probability at least $\frac{\gamma}{|X_j'|}$ (note that
$X_j'$ and $X_j$ have the same cardinality). 
Conditioned on this event, $Y_u$ will be equal to $O_u$ with probability $1/2$.
Now suppose $x \in c(X_j^{near})$. Lemma~\ref{lem:osample} implies that $O_u$ is an element of $X_j^{far}$ with probability at least $\frac{\eps}{8 \alpha t}$.
Conditioned on this event, $Y_u$ will be equal to $x$ with probability at least $\frac{1}{2} \cdot \frac{1}{|c(X_j^{near})|}$. 
Therefore, the probability that $O_u$ is equal to $x$ is at least $\frac{\eps}{8 \alpha t} \cdot \frac{1}{2|c(X_j^{near})|} \geq \frac{\eps}{16 \alpha t |X_j'|} \geq \frac{\gamma'}{|X_j'|}$.
\qed
\end{proof}

Corollary~\ref{cor:osample} shows that we can obtain samples from $X_j'$ which are nearly uniform (up to a constant factor).
To convert this to a set of uniform samples, we use the idea of~\cite{jks}.
For an element $x \in X_j'$, let $\gamma_x$ be such that $\frac{\gamma_x}{|X_j'|}$ denotes the probability that the random variable $Y_u$ is equal to $x$ (note that this is independent of $u$).
Corollary~\ref{cor:osample} implies that $\gamma_x \geq \gamma'$.
We define a new set of independent random variables $Z_1, \ldots, Z_{\eta}$.
The random variable $Z_u$ will depend on $Y_u$ only.
If $Y_u$ is $\nl$, $Z_u$ is also $\nl$.
If $Y_u$ is equal to $x \in X_j'$, then $Z_u$ takes the value $x$ with probability $\frac{\gamma'}{\gamma_x}$, and $\nl$ with the remaining probability.
\lv{Note that $Z_u$ is either $\nl$ or one of the elements of $\ostp{\ti}$.
Further, conditioned on the latter event, it is a uniform sample from $\ostp{\ti}$.}
We can now prove the key lemma.

\begin{lemma}\label{lem:final}
Let $\eta$ be $\frac{256}{\gamma' \cdot \eps}$, and $m''$ denote the mean of the non-null samples from $Z_1, \ldots, Z_{\eta}$. Then, with probability at least $(3/4)$,
$\Phi(m'', X_j') \leq (1+ \frac{\veps}{16}) \cdot \Delta(X_j')$.
\end{lemma}

\begin{proof}
Note that a random variable $Z_u$ is equal to a specific element of $X_j'$ with probability equal to $\frac{\gamma'}{|X_j'|}$.
Therefore, it takes $\nl$ value with probability $1-\gamma'$.
Now consider a different set of iid random variables $Z_u'$, $1 \leq u \leq \eta$ as follows: each $Z_u$ tosses a coin with probability of Heads being $\gamma'$.
If we get Tails, it gets value $\nl$, otherwise it is equal to a random element of $X_j'$. 
It is easy to check that the joint distribution of the random variables $Z_u'$ is identical to that
of the random variables $Z_u$.
Thus, it suffices to prove the statement of the lemma for the random variables $Z_u'$.

Now we condition on the coin tosses of the random variables $Z_u'$.
Let $n'$ be the number of random variables which are not $\nl$.
($n'$ is a deterministic quantity because we have conditioned on the coin tosses).
Let $m''$ be the mean of such non-$\nl$ variables among $Z_1', \ldots, Z_{\eta}'$.
If $n'$ happens to be larger than $\frac{128}{\veps}$, Lemma~\ref{lemma:inaba} implies that with probability at least $(7/8)$,
$\Phi(m'', X_j') \leq (1+ \frac{\veps}{16}) \cdot \Delta(X_j')$.

Finally, observe that the expected number of non-$\nl$ random variables is $\gamma' \cdot \eta \geq \frac{256}{\veps}$.
Therefore, with probability at least $\frac{7}{8}$ (using Chernoff-Hoeffding), the number of non-$\nl$ elements will be at least $\frac{128}{\veps}$.
\qed
\end{proof}

Let $C^{(\eta)}$ denotes the multi-set obtained by taking $\eta$ copies of each of the centers in $C$. 
Now observe that all the non-$\nl$ elements among $Y_1, \ldots, Y_{\eta}$ are elements of $\{O_1, \ldots, O_{\eta}\} \cup C^{(\eta)}$, and so the same must hold for $Z_1, \ldots, Z_{\eta}$. 
Moreover, since we only need a uniform subset of size $\frac{128}{\veps}$, $C_j'$ suffices instead of $C^{(\eta)}$.
This implies that in steps 5-6 of the algorithm, we would have tried adding the point $m''$ as described in Lemma~\ref{lem:final}. 
This means that $M_j$ contains a subset $T_j$ such that $\Phi(\mu(T_j), X_j) \leq (1+\frac{\veps}{2})\cdot \Delta(X_j)$ with probability at least $3/4$.
This concludes the proof of Theorem~\ref{thm:2}.

\section{Proof of Theorem~\ref{thm:weakdel}}\label{app:stab}

We restate the theorem before the proof. 

\begin{theorem}
There are instances $(X,k)$ of the $k$-means problem which are 
$\gamma$-distributed, but are not $(1+\Omega(\gamma))$-weak deletion stable for some parameter $\gamma > 0$. 
\end{theorem}

\begin{proof}
The construction of the input instance $X$ is very simple. 
There are  $n$ points in the $(n+1)$-dimensional Euclidean space (assume $n$ is even). 
The $i^{th}$ point is denoted by $p^i$ has the following coordinates: $p^i_{j} = 0$ if $j \neq i, n+1$. 
The coordinate $p^i_i$ is 1 and $p^i_{n+1}$ is $\eps$ if $i \leq n/2$, otherwise it is $-\eps$, where $\eps$ is small positive parameter. 
Let $X_1$ denote the set of points $\{p^1, \ldots, p^{n/2}\}$, i.e., the points for which coordinate $n+1$ is $\eps$, and $X_2$ be the remaining points (for which this coordinate is $-\eps$). 

We choose $k=2$. We will first verify that this solution is $(1/2)$-distributed. 
To show this, we need to figure out the structure of an optimal solution. 

\begin{claim}
\label{cl:strict}
Any optimal solution to the instance $X$ with $k=2$ must partition the set $X$ into $X_1$ and $X_2$. 
\end{claim}

\begin{proof}
Consider a solution which partitions $X$ into two clusters $C_1$ and
$C_2$ as follows, where $C_1$ contains $n_1$ points and $C_2$ contains $n_2$ points. 
Let $\mu^1$ and $\mu^2$ be the means of these clusters respectively. 
We first compute the cost of objective function corresponding to all the coordinates except for coordinate $n+1$, i.e., 
$$ \sum_{r=1}^2 \sum_{p^i \in C_r} \sum_{j=1}^n (p^i_j - \mu^r_j)^2. $$
Fix a point $p^i \in C_r$. It is easy to check that the sum of terms 
above involving $p^i$ is equal to 
$$ (1 - 1/n_r)^2 + \frac{n_r-1}{n_r^2} = 1 - 1/n_r. $$
Therefore the above sum over all points is equal to 
$$ \sum_{r=1}^2 n_r (1 - 1/n_r) = n - 2, $$
which is independent of the clustering. However, the contribution towards
the objective function corresponding to coordinate $n+1$ is 0 if and only if the clusterings are $X_1$ and $X_2$. This proves the claim. 
\qed
\end{proof}
Now let is check the $\gamma$-distributed property. By the proof of the claim above, we know that optimal value is $n-2$, and each of the optimal clusters have size $n/2$. Now, consider a point $p^i$, and assume wlog that $i \leq n/2$. So $p^i$ belongs to the optimal cluster $X_1$. Note that $\mu^2$ has coordinates given by $$ \mu^2_j = \left\{ 
\begin{array}{cc} 0 & \mbox{ if $ j \leq n/2$} \\
2/n & \mbox { if $n/2 < j \leq n$} \\
-\eps & \mbox{ if $j = n+1$} \end{array} \right. $$
Therefore, $$ ||p^i - \mu^2||^2 = 1 + 4/n + \eps^2 \geq \frac{1}{2} \cdot\frac{OPT}{n/2}.$$
Therefore, the instance is $1/2$-distributed. 
Now we check that this instance is not $(1+O(1))$-weakly deletion stable. 
Suppose we assign all the points in the cluster $X_1$ to $\mu_2$.
The cost of this clustering, where all points are now getting assigned to $\mu_2$ can again be easily calculated. It is easy to check that this quantity is $n(1 + 2 \eps^2)$. The ratio of the increased cost to the optimal cost is $$\frac{n(1+2\eps^2)}{n-2} = 1 + \theta(1/n). $$
Therefore, this instance is not $\Omega(1)$-weakly deletion stable. This completes the proof of Theorem~\ref{thm:weakdel}. \qed
\end{proof}

\end{document}